\documentclass[twoside,11pt]{article}
\usepackage{nonatbib}
\usepackage[T1]{fontenc}
\usepackage{microtype}
\usepackagewithoutnatbib{jmlr2e}
\usepackage{fixthmtools}
\usepackage{hyperref}
\usepackage{lastpage}
\usepackage{subcaption}

\usepackage[authoryear,citations,centerfigures,commands,enumerate,environments,lmrscinnershape,theorems]{AVT}

\bibliography{references}
\let\UrlFont\scshape

\usepackage{tikz}

\renewcommand{\Re}{\operatorname{Re}}

\DeclareMathOperator{\GL}{GL}
\DeclareMathOperator{\Ort}{O}
\DeclareMathOperator{\SL}{SL}
\DeclareMathOperator{\SO}{SO}
\DeclareMathOperator{\SU}{SU}

\DeclareMathOperator{\SPD}{SPD}

\DeclareMathOperator{\Gr}{Gr}
\DeclareMathOperator{\M}{M}
\DeclareMathOperator{\RP}{RP}

\declaretheorem[style=plain]{Assumption}

\Crefname{equation}{}{}

\title{Stationary Kernels and Gaussian Processes on Lie Groups and their Homogeneous Spaces I: the compact case}

\author{\name Iskander Azangulov\textsuperscript{\ensuremath{*}}
\email iska.azn@gmail.com \\
\addr St. Petersburg State University and University of Oxford
\AND
\name Andrei Smolensky\textsuperscript{\ensuremath{*}}
\email andrei.smolensky@gmail.com \\
\addr St. Petersburg State University and Neapolis University Pafos
\AND
\name Alexander Terenin
\email avt28@cornell.edu \\
\addr University of Cambridge and Cornell University
\AND
\name Viacheslav Borovitskiy
\email viacheslav.borovitskiy@gmail.com \\
\addr ETH Zürich
}

\jmlrheading{25}{2024}{1-\pageref{LastPage}}{12/22; Revised
11/23}{3/24}{22-1434}{Iskander Azangulov, Andrei Smolensky, Alexander Terenin, Viacheslav Borovitskiy}
\ShortHeadings{Stationary Kernels and Gaussian Processes on Lie Groups and Homogeneous Spaces I}{Azangulov, Smolensky, Terenin, Borovitskiy}

\editor{Philipp Hennig}

\begin{document}

\maketitle

\begin{table}[b!]
\vspace*{-1.5ex}
\footnoterule
\footnotesize\textsuperscript{\ensuremath{*}}Joint first author.
\hfill\null
\end{table}

\begin{abstract}
Gaussian processes are arguably the most important class of spatiotemporal models within machine learning.
They encode prior information about the modeled function and can be used for exact or approximate Bayesian learning.
In many applications, particularly in physical sciences and engineering, but also in areas such as geostatistics and neuroscience, invariance to symmetries is one of the most fundamental forms of prior information one can consider.
The invariance of a Gaussian process' covariance to such symmetries gives rise to the most natural generalization of the concept of stationarity to such spaces.
In this work, we develop constructive and practical techniques for building stationary Gaussian processes on a very large class of non-Euclidean spaces arising in the context of symmetries.
Our techniques make it possible to (i) calculate covariance kernels and (ii) sample from prior and posterior Gaussian processes defined on such spaces, both in a practical manner.
This work is split into two parts, each involving different technical considerations: part I studies compact spaces, while part II studies non-compact spaces possessing certain structure.
Our contributions make the non-Euclidean Gaussian process models we study compatible with well-understood computational techniques available in standard Gaussian process software packages, thereby making them accessible to practitioners.
\end{abstract}

\begin{keywords}
Gaussian processes, kernels, geometric learning, stationarity, symmetries, Lie groups, homogeneous spaces, Riemannian manifolds.
\end{keywords}

\section{Introduction}

Gaussian processes are widely used as machine learning models for learning unknown functions, particularly when data is scarce.
Using Bayesian methods, Gaussian processes are able to quantify uncertainty, thereby allowing what is known and what is unknown to be balanced when modeling and making decisions.
To maximize statistical efficiency, and ensure their posterior error bars are appropriate and can be trusted, the Gaussian process models we employ in practice should be adapted to their respective settings.
Their prior information should correctly encode the structure of the problem at hand.

Many problems, particularly in the physical sciences and engineering, require modeling functions on non-Euclidean spaces.
For example, spatial statistics on the surface of earth \cite{jeong2017,guinness2016}, cosmology \cite{lang2015} and medical imaging \cite{andersson2015} often call for Gaussian process models on the sphere, which is a very well-studied setting---see for instance \textcite{marinucci2011}.
Other applications, however, require Gaussian process models on different spaces: for many of these, appropriate statistical tools are only partially developed, particularly if the application demands careful consideration of symmetries.
For example, control parameter tuning  and parametric policy adaptation in robotics \cite{jaquier2020, jaquier2022}, latent representation learning in neuroscience \cite{jensen2020} or modeling unknown dynamics of physical systems \cite{borovitskiy2020,hutchinson2021} require models on various \emph{Lie groups} and their \emph{homogeneous spaces}.
At present, these applications tend to rely on ad-hoc heuristic or partially understood techniques, rather than on a principled and comprehensive formal treatment.

In this two-part work, we study Gaussian process priors on a large class of non-Euclidean spaces.
We focus on \emph{stationary} priors, whose covariance kernels are invariant under symmetries of the space, generalizing the usual Euclidean notion of stationarity as shift-invariance.
We work with spaces whose group of symmetries is rich enough, including, in part I, compact \emph{Lie groups} such as the special orthogonal group $\SO(n)$ or the special unitary group $\SU(n)$, and their \emph{homogeneous spaces} such as the sphere $\bb{S}_n$ or the Stiefel manifold $\f{V}(k, n)$.
In part~II~\cite{part2}, we consider certain \emph{non-compact symmetric spaces} such as the hyperbolic spaces $\bb{H}^n$ and spaces of symmetric positive definite matrices~$\SPD(n)$.
As part of our work, we contribute implementations of these ideas to the \href{https://geometric-kernels.github.io}{\textsc{GeometricKernels}}\footnote{Available at \url{https://geometric-kernels.github.io}. A prototypical software implementation, as well as our experiments, are can be found at \url{https://github.com/imbirik/LieStationaryKernels}.} library.
We now proceed to specify notation and more formally introduce the setting.

\paragraph*{Notation.}
We use bold italic letters ($\v{a}$, $\v{b}$) to denote finite-dimensional vectors in $\R^n$, as well as elements of the product space $X^n$ for a given set $X$.
We use bold upface letters ($\m{A},\m{B}$) to denote matrices.
If $f: X \-> Y$ then for $\v{x} \in X^n$ we write $f(\v{x})$ to denote the vector $\del{f(x_1), \ldots, f(x_n)}^{\top}$.
Analogously, for a function~$k: A \x B \-> C$, for $\v{a} \in A^n$ and $\v{b} \in B^m$ we write $k(\v{a}, \v{b})$ to denote the matrix with entries $k(a_\ell, b_{\ell'})$ for $1 \leq \ell \leq n$ and $1 \leq \ell' \leq m$.
We often write this as $\m{K}_{\v{a} \v{b}} = k(\v{a}, \v{b})$.
We use $\conj{x}$ to denote conjugation of complex numbers.
We denote the conjugate transpose of a matrix by $\m{A}^* = \conj{\m{A}^{\top}}$.

\subsection{Kernels and Gaussian Processes}
\label{sec:gp-intro}

Let $X$ be a set.
We say that a stochastic process $f$ indexed by $X$ is a \emph{Gaussian process} if, for any finite set of inputs $\v{x} \in X^N$, the random vector $f(\v{x})$ is multivariate Gaussian.
Every such process is determined by its \emph{mean function} $\mu(\.) = \E(f(\.))$ and \emph{covariance kernel} $k(\.,\.') = \Cov(f(\.),f(\.'))$, so we write ${f \~[GP](\mu, k)}$.
We assume all Gaussian processes we consider are real-valued, unless stated otherwise or apparent from context.

Assume $y_i = f(x_i) + \eps_i$ with $\v\eps \~[N](\v{0},\m\Sigma)$, and let $\v{x},\v{y}$ be the data. 
If $f\~[GP](0,k)$ is the prior Gaussian process, then the posterior $f\given\v{y}$ is also a Gaussian process, whose mean and covariance kernel are
\[
\E(f\given\v{y}) &= \m{K}_{(\.)\v{x}}(\m{K}_{\v{x}\v{x}} + \m\Sigma)^{-1}\v{y}
&
\Cov(f\given\v{y}) &= \m{K}_{(\.,\.')} - \m{K}_{(\.)\v{x}}(\m{K}_{\v{x}\v{x}} + \m\Sigma)^{-1}\m{K}_{\v{x}(\.')}
\]
where $(\.)$ and $(\.')$ denote arbitrary sets of evaluation locations.
We will also use another way of thinking about the posterior: using \emph{pathwise conditioning} \cite{chiles2009,wilson2020,wilson2021}.
From this viewpoint, the posterior can be written
\[
\label{eqn:intro:pathwise}
(f\given \v{y})(\.) = f(\.) + \m{K}_{(\.)\v{x}} (\m{K}_{\v{x}\v{x}} + \m\Sigma)^{-1}(\v{y} - f(\v{x}) - \v\eps)
&&
\v\eps \~[N](\v{0},\m\Sigma)
\]
where equality is taken to be in distribution.

In order to perform Bayesian learning with Gaussian processes, a prior Gaussian process, or a parametric family thereof, must be chosen.
In the Euclidean setting, priors given by \emph{stationary kernels} are a widely-used choice, which we now discuss.

\subsection{Stationary Kernels and Gaussian Processes in Euclidean Spaces}

A kernel $k$ over $\R^n$ is called \emph{stationary} if it is translation invariant in the sense that
\[
k(\v{x} + \v{c}, \v{x}' + \v{c}) = k(\v{x}, \v{x}')
\]
for all $\v{c} \in \R^n$.
In this case, there is a function $\Bbbk : \R^n \-> \R$ such that
\[
k(\v{x},\v{x}') = \Bbbk(\v{x}-\v{x}')
.
\]
A Gaussian process, similarly, is called stationary if its mean function is translation invariant---hence, constant, usually assumed zero---and its covariance kernel is stationary.

Stationary kernels are widely used in applications.
This includes the widely-popular family of \emph{Matérn} kernels advocated by \textcite{stein1999}, which uses three interpretable parameters: amplitude $\sigma^2$ which controls variability, length scale $\kappa$ which controls spatial decay, and smoothness $\nu$ which controls differentiability of the process.
The limiting case as $\nu\->\infty$ of the Matérn family is arguably the single most popular kernel used, the squared exponential (heat, Gaussian, RBF, diffusion) kernel which induces an infinitely differentiable Gaussian process prior.
We will study these kernels in more general settings: as a glimpse of what we will be able to carry out, we illustrate Gaussian process regression with a Matérn kernel on a real projective space, using the techniques of \Cref{sec:heat_matern}, in \Cref{fig:gp-regression}.

\begin{figure}
\begin{subfigure}{0.2425\textwidth}
\includegraphics[scale=0.25]{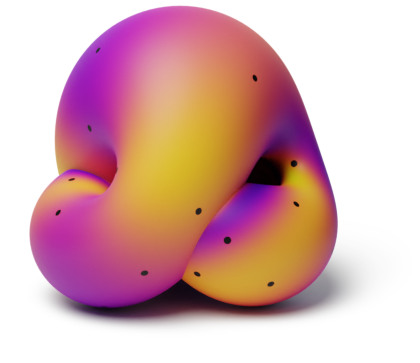}
\caption{Ground truth}
\end{subfigure}
\begin{subfigure}{0.2425\textwidth}
\includegraphics[scale=0.25]{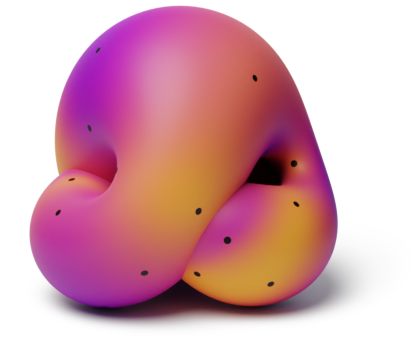}
\caption{Mean}
\end{subfigure}
\begin{subfigure}{0.2425\textwidth}
\includegraphics[scale=0.25]{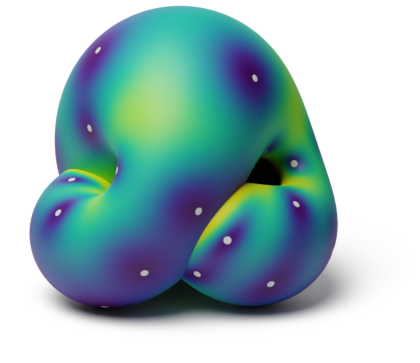}
\caption{Standard dev.}
\end{subfigure}
\begin{subfigure}{0.2425\textwidth}
\includegraphics[scale=0.25]{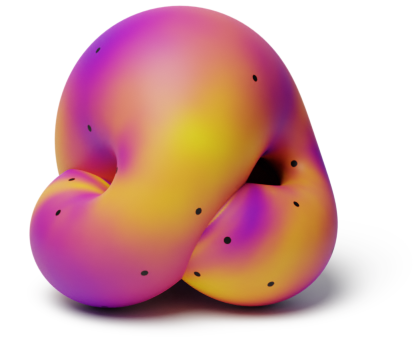}
\caption{Posterior sample}
\end{subfigure}
\caption{We illustrate a regression problem on a real projective plane $\RP_2$, with a Matérn-3/2 kernel in the sense of \Cref{sec:heat_matern}. Here and below we use the Boy's surface immersion of $\RP_2$ into $\mathbb{R}^3$ \cite{bryant1988}, which is non-isometric. Regression is applied to model the function  $y(x) = \sum_{j=1}^5 \cos(d_{\RP_2}(x, \xi_j))$ where $\xi_j$ are uniform samples on $\RP_2$. 
We show the resulting posterior mean and standard deviation, and one posterior sample.}
\label{fig:gp-regression}
\end{figure}

Stationary kernels have important properties.
Bochner's Theorem says that continuous stationary kernels are in bijective correspondence with finite nonnegative measures via the inverse Fourier transform, namely
\[
\Bbbk(\v{x}-\v{x}')
=
\int_{\R^n} e^{2\pi i \innerprod{\v\psi}{\v{x}-\v{x}'}} \d\Psi(\v\psi)
=
\int_{\R^n} \cos\del[0]{2\pi \innerprod[0]{\v\psi}{\v{x}-\v{x}'}} \d\Psi(\v\psi)
\]
where $\Psi$ is the spectral measure of $\Bbbk$, and the second equality follows from $\Bbbk$ being real-valued.
This both gives a way of defining new kernels---by specifying their respective spectral measures---and calculating them numerically.
In particular, defining $\v\phi$ to be a vector of complex exponentials, and, following \textcite{rahimi08}, discretizing the integral by using Monte Carlo, we obtain the approximation
\[
\Bbbk(\v{x}-\v{x}') \approx \v\phi(\v{x})^\top\v\phi(\v{x}').
\]
Letting $\phi_s$ be the components of $\v\phi$, this further leads to the expression
\[ \label{eqn:intro_f_rff}
f(\.) \approx \sum_{s=1}^S w_s \phi_s(\.)
&
&
w_\ell \~[N](0,1)
\]
for approximating the Gaussian process itself.
This enables one to efficiently draw approximate samples from the prior, and transform them into approximate posterior samples via \Cref{eqn:intro:pathwise} by plugging in the approximate prior in place of $f(\.)$ and $f(\v{x})$.

The \emph{feature map} $\v\phi$ is not uniquely determined: one choice is to use sines and cosines
\[
\label{eq:euclidean:rff}
\phi_{2 s - 1} (\v{x})
&=
\sigma \sqrt{\frac{2}{S}} \cos(2\pi \innerprod{\v\psi_s}{\v{x}})
&
\phi_{2 s} (x)
&=
\sigma \sqrt{\frac{2}{S}} \sin(2\pi \innerprod{\v\psi_s}{\v{x}})
&
\v\psi_s &\~\frac{1}{\sigma^2}\Psi
\]
where $s \in {1, \ldots, S/2}$ for $S$ even, and $\sigma^2 = \Psi(\R^n)$.
Observe that, here, sines and cosines come in pairs: each sample from the spectral measure corresponds to exactly two functions used in the approximation.
An alternative choice is to use only cosines, but add a random phase
\[ \label{eqn:intro_rfff}
\phi_s(\v{x}) &= \sigma \sqrt{\frac{2}{S}} \cos(2\pi \innerprod{\v\psi_s}{\v{x}} + \beta_s)
&
\v\psi_s &\~ \frac{1}{\sigma^2}\Psi
&
\beta_s &\~[U](0,2\pi)
\]
for $s \in 1, \ldots, S$ and $\sigma^2 = \Psi(\R^n)$, where $\f{U}$ is the uniform distribution.
An important advantage of $\phi$ in \Cref{eq:euclidean:rff,eqn:intro_rfff} over the complex exponentials we begun with is their real-valuedness.
General analogs of Bochner's theorem and Fourier feature expansions similar to~\Cref{eq:euclidean:rff,eqn:intro_rfff} will be the cornerstone of our considerations in the geometric settings~we~study.

\subsection{Previous Work and Contributions}

Recently, authors in the statistics and machine learning communities have begun exploring generalizations of Gaussian process models to non-Euclidean settings---and, in particular to Riemannian manifolds, of which the spaces studied in this work are special cases.

The first difficulty here occurs in defining a positive semi-definite covariance kernel to begin with.
For instance, \textcite{feragen15} have shown that, for a Riemannian manifold $X$ with geodesic distance $d$, if one generalizes the widely used squared exponential kernel naïvely by substituting the Euclidean distance with $d$, then the obtained expression
\[ \label{eqn:geodesic_sq_exp}
\sigma^2 \exp\del{-\frac{d(x,x')^2}{2\kappa^2}}
\]
will be positive semi-definite for all $\kappa > 0$ if and only if $X$ is isometric to a Euclidean space.
In some cases, there exist even stronger results of this type.
For example, \textcite{dacosta2023} show that the kernel in~\Cref{eqn:geodesic_sq_exp} is not positive semi-definite for all values of $\kappa$---rather than for some---if the space $X$ of interest belongs to the class of \emph{symmetric spaces},\footnote{The class of symmetric spaces is a subclass of homogeneous spaces that includes, for example, spheres and hyperbolic spaces.} thereby solving the open problem posed by \textcite{feragen16} in the case of symmetric spaces.
This distance-based approach is therefore a technical disaster in need of alternatives.

One alternative, pursued by \textcite{lindgren11,said13,coveney20,borovitskiy2020} and mostly restricted to compact manifolds, involves defining Riemannian Gaussian processes using stochastic partial differential equations.
A key result in this setting is that the kernels of certain such processes admit \emph{manifold Fourier feature} expansions
\[ \label{eqn:compact_manifold_kernel}
k(x,x') = \sigma^2 \sum_{j=0}^\infty \Psi(\lambda_j) f_j(x) f_j(x')
\]
where $(\lambda_j,f_j)$ are eigenpairs of the (negated) Laplace--Beltrami operator $-\Delta$, and $\Psi(\lambda_j)$ can be viewed as the generalized spectral measure---now supported on the nonnegative integers.

The upside of this approach is that, in principle, it allows one to compute covariance kernels without needing to solve stochastic partial differential equations.
The downside is that computing the eigenpairs analytically or numerically is for many manifolds a non-trivial problem in its own right.
One can observe that~\Cref{eqn:compact_manifold_kernel} coincides with various explicit expansions known for spheres and tori, derived from Bochner-type symmetry-based considerations, rather than from differential equations.
This suggests that to build a deeper understanding, one should also study these phenomena from the perspective of stationarity under group action. 
This will be the starting point for our work.
For instance, if viewed through an appropriate lens, in part I of this work we present techniques to bypass or solve the problem of finding Laplace--Beltrami eigenpairs for a large class of compact Riemannian~manifolds.

If we set $S(\lambda) = e^{-\lambda t}$ in~\Cref{eqn:compact_manifold_kernel}, the obtained expression coincides with the \emph{heat kernel} (also known as the \emph{diffusion kernel}), the fundamental solution of the heat equation, which generalizes the squared exponential kernel in a sound manner, without incurring the problems of geodesic-based approaches such as \Cref{eqn:geodesic_sq_exp} \cite{grigoryan2009}.
This kernel has been studied by many authors in different specific settings \cite{lafferty2005,kondor2008,zhao2018}, and is a landmark target of development for our work.
\textcite{niu2019,ye2020} suggest a general approach to approximate heat kernel's values by simulating random walks, leveraging connections between the heat equation and Brownian motion.
While this approach can work, it is also a computationally-expensive method-of-last-resort, since it can converge slowly and is not guaranteed to give a positive semi-definite kernel for finite samples.
We develop different techniques that simultaneously apply to more general kernels, including analogs of Matérn kernels, and obtain approximations which are guaranteed positive semi-definite.

We begin with an abstract description of stationary Gaussian processes on a Riemannian manifold $X$ acted on by a suitably well-behaved Lie group $G$, given by \textcite{yaglom1961} at the dawn of stochastic process theory.
This description relies on certain far-reaching generalizations of harmonic analysis built upon the notions of \emph{representation theory}.
We develop this abstract description into a fully-accessible computational tool suitable for day-to-day statistical practice.
We do this by making Yaglom's abstract expressions constructive: specifically, we bring together and synthesize various theory and techniques to develop computational methods for statistical learning based on representation-theoretic expressions.
This includes the following classes of approximations.

\1 Algorithms for pointwise kernel evaluation, with automatic differentiation support.
\2 Algorithms for efficiently drawing sample paths, both from priors and posteriors.
\0

These algorithms approximate stationary kernels by linear combinations of functions which depend only on the given space, and admit closed-form expressions arising from algebraic considerations.
The coefficients in these linear combinations depend on the particular kernel, but for the aforementioned heat and Matérn kernels can be computed exactly.
Our sampling routines rely on the same quantities, and additionally require only that one can sample uniformly on the space of interest.

Together, these algorithmic primitives provide practitioners with all of the tools needed to define Gaussian processes for spatial models \cite{cressie92} and decision-making algorithms such as Bayesian optimization \cite{frazier18}, along with primitives for constructing vector-valued processes \cite{hutchinson2021} and other extensions, and studying their theoretical properties \cite{rosa23}.
The techniques developed provide a blueprint for how one might derive other non-Euclidean kernels from symmetry-based considerations, including potentially on discrete spaces \cite{borovitskiy2021,borovitskiy2022} or more general algebraic structures where representation-theoretic considerations apply \cite{fukumizu08}.
This is a purely theoretical paper: for applications, we refer the reader to the aforementioned~works.

In total, our work makes it possible to use Gaussian processes on a much wider collection of manifolds admitting analytic descriptions than was previously~practically~available.

\section{Stationary Gaussian Processes under Group Action} \label{sec:stationary}

Let $X$ be a set, and let $G$ be a group acting on $X$ with group action $\lacts : G \x X \-> X$.
A kernel is called \emph{stationary}\footnote{Note that the term \emph{homogeneous} is also used in the literature, notably by \textcite{yaglom1961}.} if it satisfies 
\[
k(g \lacts x, g \lacts x') = k(x,x')  
\]
for all $g \in G$ and all $x,x' \in X$.
Unless otherwise stated, we assume that $k$ is real valued and continuous.
From this, one can easily see that if $f\~[GP](0,k)$ has a stationary kernel, then the processes $f(g\lacts \cdot)$ and $f(\cdot)$ have the same distribution.
In this case, we say that $f$ is stationary.
Without loss of generality, we work with processes with zero mean.
We now introduce the key results of \textcite{yaglom1961} for different classes of spaces---these are, effectively, far-reaching generalizations of the Bochner's theorem.

\subsection{Stationary Kernels on Lie Groups} \label{sec:stationary:lie}

A \emph{Lie group} $(G,\bdot)$ is a group which is also a smooth manifold.
We focus on the case where $G$ acts on itself from both the left and the right, and defer the case where only one of these actions is present to \Cref{sec:stationary:homogeneous}.
This falls into the general notion of stationarity introduced earlier by considering the product group $G \times G$ acting on the set $G$ by $(g_1, g_2) \lacts x = g_1 \bdot x \bdot g_2^{-1}$.
In that case, for $(g_1, g_2) \in G \times G$, we have
\[
k((g_1, g_2) \lacts x, (g_1, g_2) \lacts x')
=
k(g_1 \bdot x \bdot g_2^{-1}, g_1 \bdot x' \bdot g_2^{-1})
=
k(x,x')
.
\]
Kernels---and general functions---that satisfy this property are called \emph{bi-invariant}.
Letting $e \in G$ be the identity element, we can define a single-argument function $\Bbbk : G \-> \R$ such~that
\[
k(x_1,x_2)
=
k(x_2^{-1} \bdot x_1, e)
=
\Bbbk(x_2^{-1} \bdot x_1)
=
\Bbbk(x_1 \bdot x_2^{-1})
\]
which directly generalizes the one-argument notion of the Euclidean case.

We will need notions and results from the \emph{representation theory} of compact Lie groups.
The idea of representation theory is to study the structure of a group $G$ by studying how it can be instantiated as the group of transformations of some vector space, which is often more concrete and easier to understand.
More formally, a \emph{representation} $\pi : G \-> \GL(V)$ is a homomorphism from $G$ to the general linear group, consisting of bijective linear transformations of a (complex) vector space~$V$.
A representation $\pi$ is called \emph{unitary} if its range consists of unitary matrices.
Representations that do not have any invariant subspaces except $\{0\}$ and $V$, and therefore cannot be decomposed into a direct sum of smaller representations, are called \emph{irreducible}.

We work with smooth representations of compact Lie groups.
In the setting at hand, there are at most a countable number of irreducible unitary representations---let $\Lambda$ be an index set in bijective correspondence with the set of such representations.
We defer the question of how to actually obtain such a set to subsequent sections.
Let $\lambda\in\Lambda$ be an index.
One can show that every smooth irreducible unitary representation
\[
\pi^{(\lambda)} : G \-> GL(V_\lambda)
\]
is finite-dimensional, in the sense that the vector space $V_\lambda$ is finite-dimensional.
Let $d_\lambda = \dim V_\lambda$.
If we introduce a $G$-invariant inner product $\innerprod{\cdot}{\cdot}_{V_{\lambda}}$ and an orthonormal basis $e_1, \ldots, e_{d_{\lambda}}$ in $V_{\lambda}$, we can define the \emph{matrix coefficients} of the representation $\pi^{(\lambda)}$ by
\[
\pi_{j k}^{(\lambda)}: G &\-> \C
&
\pi_{j k}^{(\lambda)}(g) &= \innerprod[1]{\pi^{(\lambda)}(g) e_j}{e_k}_{V_{\lambda}}
\]
where $1 \leq j,k \leq d_\lambda$.
Let $\mu_G$ be the unique probabilistic Haar measure on $G$, which by compactness always exists.
By the Peter--Weyl Theorem---see \textcite{folland1995} as well as \textcite{weil1940}---the matrix coefficients form an orthogonal basis in the space $L^2(G, \mu_G)$ with $\norm[0]{\pi_{j k}^{(\lambda)}}_{L^2(G, \mu_G)} = 1/\sqrt{d_{\lambda}}$.
Thus, $f \in L^2(G, \mu_G)$ may be expressed in this basis in the form of an infinite series using the functions $\sqrt{d_{\lambda}}\pi_{j k}^{(\lambda)}(\.)$.
This generalizes Fourier series to the case of a compact topological group $G$.
Define the \emph{character} $\chi^{(\lambda)}: G \-> \C$ corresponding to the representation $\pi^{(\lambda)}$ by the expression
\[
\chi^{(\lambda)}(g)
=
\tr \pi^{(\lambda)}(g)
=
\sum_{j = 1}^{d_{\lambda}} \pi^{(\lambda)}_{j j}(g)
.
\]
Note that the character does not depend on a choice of a basis $e_1,\ldots,e_{d_{\lambda}}$---a property inherited from the trace operator $\tr$.
Having introduced the necessary notions, we are now ready to present the representation-theoretic description of stationary Gaussian processes on compact Lie groups.

\begin{figure}
\begin{subfigure}{0.24\textwidth}
\includegraphics[scale=0.175,trim=75 55 75 75,clip]{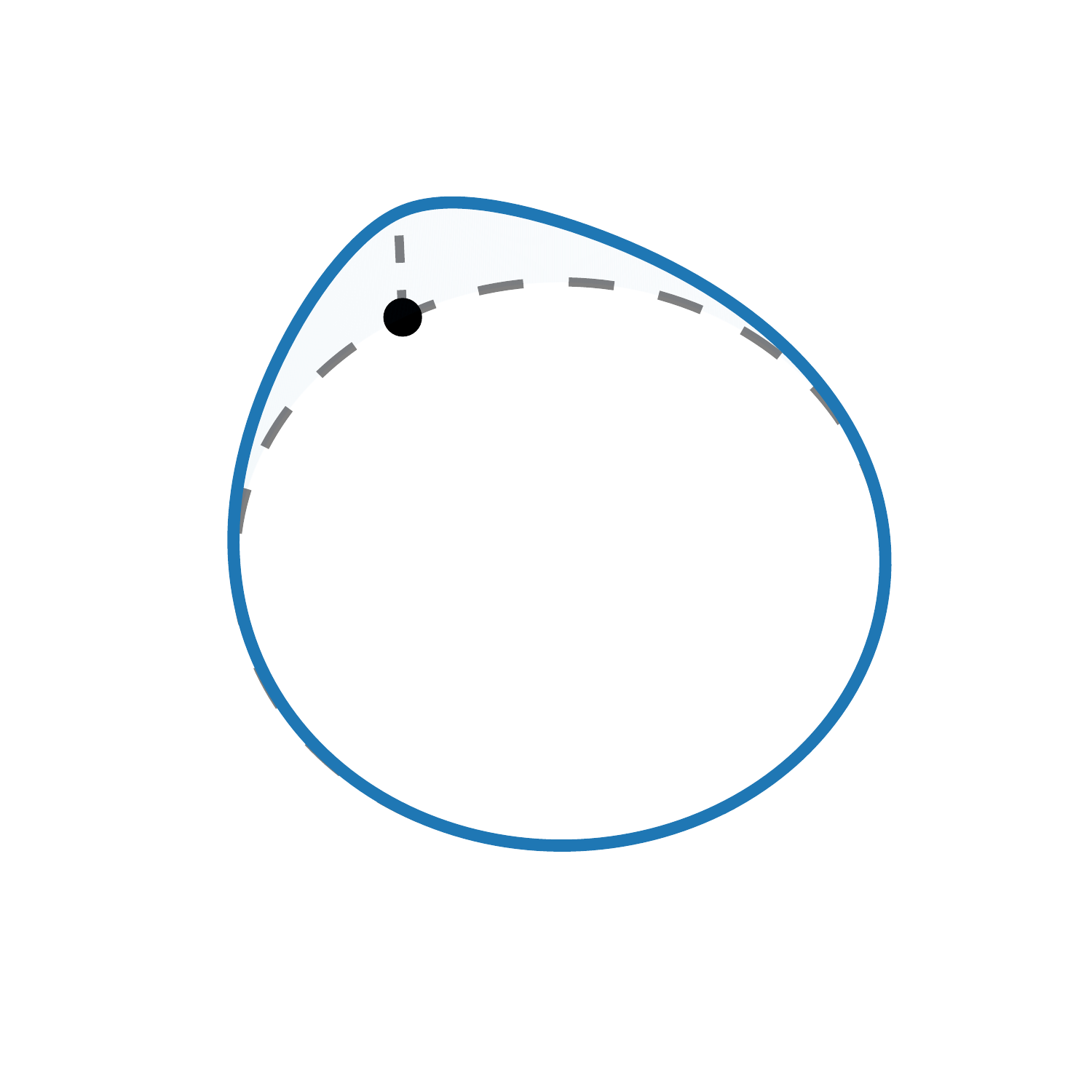}
\caption{Kernel}
\end{subfigure}
\begin{subfigure}{0.24\textwidth}
\includegraphics[scale=0.175,trim=75 55 75 75,clip]{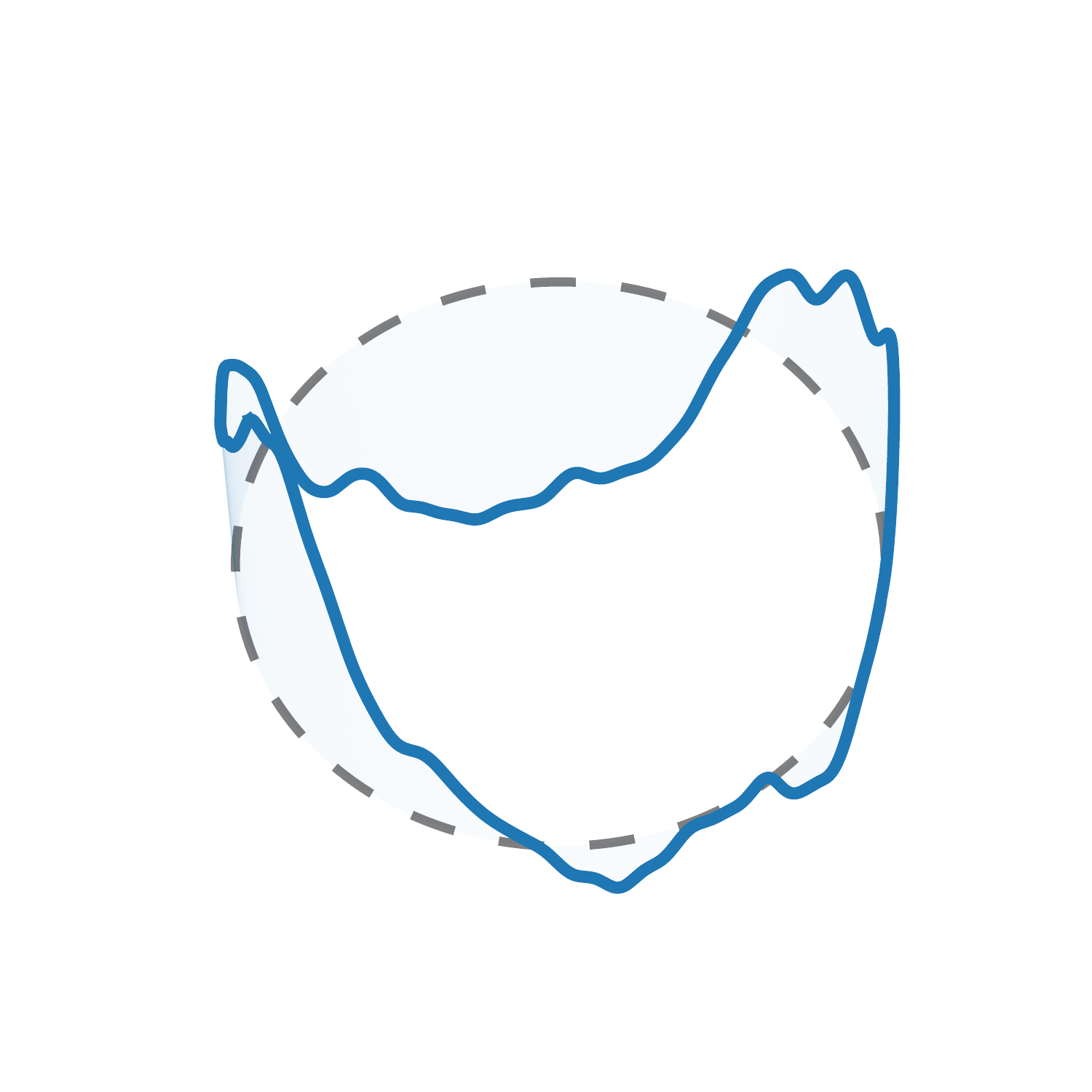}
\caption{Sample}
\end{subfigure}
\begin{subfigure}{0.24\textwidth}
\includegraphics[scale=0.175,trim=75 55 75 75,clip]{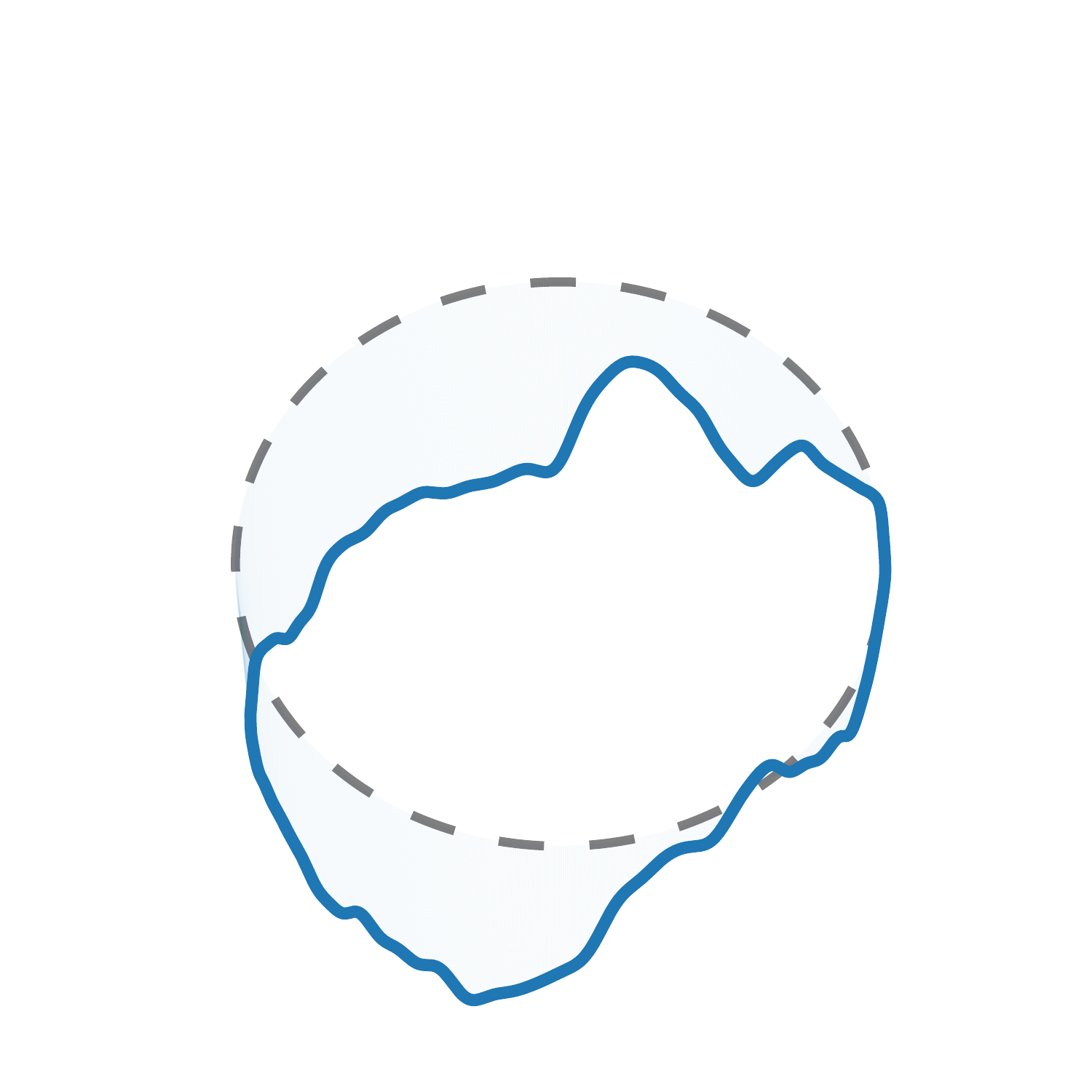}
\caption{Sample}
\end{subfigure}
\begin{subfigure}{0.24\textwidth}
\includegraphics[scale=0.175,trim=75 55 75 75,clip]{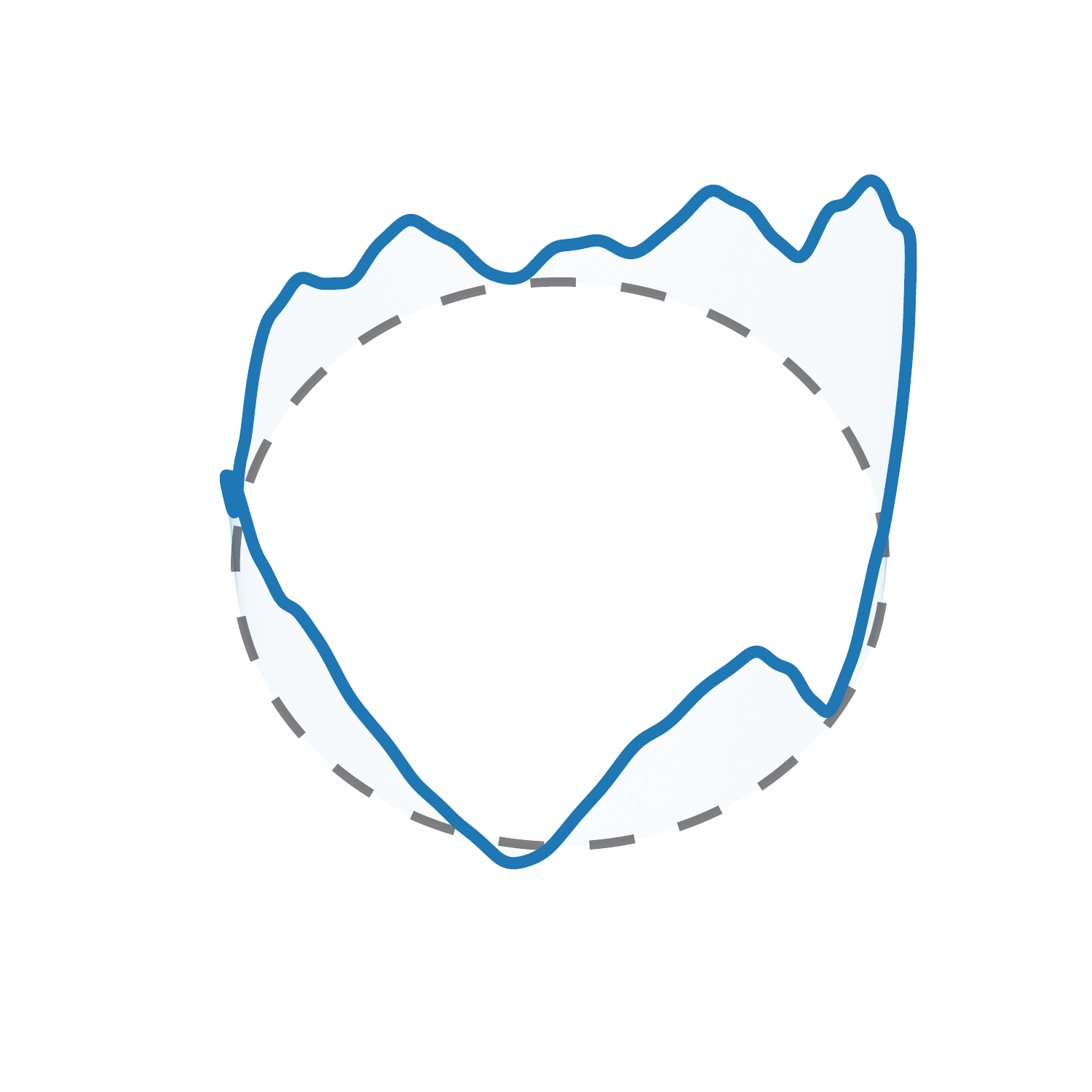}
\caption{Sample}
\end{subfigure}
\caption{Illustration of the Matérn-3/2 kernel of \Cref{sec:heat_matern}, and three random samples from its respective prior Gaussian process, on the circle group $\mathbb{S}_1$, which acts on itself by~rotation.}
\label{fig:circle-kernel}
\end{figure}

\begin{restatable}{theorem}{ThmStationaryGroup} \label{thm:stationary_group}
A Gaussian process $f \~[GP](0, k)$ on a compact Lie group $G$ is stationary with respect to left-and-right action of $G$ on itself if and only if $k$ is of form
\[ \label{eqn:stationary_lie:kernel}
k(g_1, g_2) = \sum_{\lambda \in \Lambda} a^{(\lambda)} \, \Re\chi^{(\lambda)}(g_2^{-1} \bdot g_1)
\]
where $a^{(\lambda)} \geq 0$ satisfy $\sum_{\lambda\in\Lambda} d_{\lambda} a^{(\lambda)} < \infty$.
Moreover, for all $\lambda$, $\Re\chi^{(\lambda)}$ is positive-definite.
\end{restatable}

\begin{proof}
This follows mostly by specializing Theorem 2 of \textcite{yaglom1961}, which describes complex-valued stationary processes on general topological groups, to the case where $f$ is real-valued and Gaussian.
A detailed proof is given in \Cref{appdx:proofs}.
\end{proof}

This is a powerful result: it says that to compute a general stationary kernel on a Lie group numerically, it suffices to compute the coefficients $a^{(\lambda)}$, and the characters $\chi^{(\lambda)}$, up to some order of truncation.
To carry this out, we will need some way to represent the index set $\Lambda$ of irreducible unitary representations numerically.
We describe methods for doing so and the methods for efficient approximate sampling from such processes in~\Cref{sec:computation}.
For the circle group $G = \bb{S}_1 = \SO(2) = \mathrm{U}(1)$, an illustration of such a kernel, and samples of a corresponding process, computed using the techniques of \Cref{sec:computation}, is given in \Cref{fig:circle-kernel}.
Before proceeding to these, we study the situation on other spaces of interest.

\subsection{Stationary Kernels on Homogeneous Spaces of Compact Lie Groups} \label{sec:stationary:homogeneous}

Let $G$ be a compact Lie group, and let $H$ be a closed, not necessarily normal subgroup of~$G$.
We say that $X = G / H$ is a \emph{homogeneous space}.\footnote{In geometry, a homogeneous space is a topological space $X$ upon which some group $G$ of automorphisms acts continuously and transitively \cite[Section 2.6]{folland1995}. In this case, if we fix any point $x \in X$ and define $H$ to be the subgroup of $G$ consisting of such elements $g \in G$ that $g \lacts x = x$, it can be proved that $X$ is isomorphic to $G / H$.}
$X$ consists of cosets of the form $g \bdot H$---the equivalence classes defined in the natural manner, and the action of $G$ on $X$ is defined by $g_1 \lacts (g_2 \bdot H) = (g_1 \bdot g_2) \bdot H$.
For simplicity, where unambiguous we will often drop the $(\.) \bdot H$ component from notation, and write elements of $G / H$ as $x \in G / H$.

Consider complex-valued functions $f : G/H \-> \C$ defined on the homogeneous space.
It is not hard to see that such functions are in bijective correspondence with functions $\tilde{f} : G \-> \C$ which are constant on all cosets, namely $\tilde{f}(g) = \tilde{f}(g \bdot h)$ for all $h\in H$.
Moreover, if we define the natural projection $\phi : G \-> G/H$ by $\phi : g \|> g\bdot H$, then every function $f \in L^2(G/H)$ can be written as $\tilde{f} \after \phi \in L^2(G)$.
This shows homogeneous spaces inherit square-integrability from the Lie group $G$ of their automorphisms \cite{folland1995}.

We can therefore expand a function $f \in L^2(G/H)$ using the $L^2(G)$-orthonormal bases introduced previously.
In particular, every such function can be written as an infinite weighted sum of matrix coefficients $\pi^{(\lambda)}_{jk}(\.)$.
On the other hand, not every function constructed this way will be constant on the cosets, as all functions $f \in L^2(G/H)$ must be.
To get around this, we revisit the bases chosen in the vector spaces $V_\lambda$, which were previously used to define the matrix coefficients.
For an $f \in L^2(G/H)$, let $\tilde{f}(g) = (f\after\phi)(g)$, $\tilde{f} \in L^2(G)$.
Write
\[
\tilde{f}(g)
&=
\sum_{\lambda\in\Lambda}
d_{\lambda}
\sum_{j, k = 1}^{d_{\lambda}}
\innerprod[1]{\tilde{f}}{\pi^{(\lambda)}_{j k}}_{L^2(G)} \pi^{(\lambda)}_{j k}(g)
\\
&=
\sum_{\lambda\in\Lambda}
d_{\lambda}
\int_G \tilde{f}(u)
\ubr{\del[2]{
\sum_{j, k = 1}^{d_{\lambda}}
\pi^{(\lambda)}_{j k}(g)
\overline{\pi^{(\lambda)}_{j k}(u)}
}}_{\chi^{(\lambda)}(u^{-1} \bdot g)}
\d \mu_G(u)
=
\sum_{\lambda\in\Lambda}
d_{\lambda}\,
(\tilde{f} * \chi^{(\lambda)}) (g).
\]
where $*$ denotes convolution, and the calculation follows by interchanging sums and integrals.
We therefore see that the convolution with a character $\chi^{(\lambda)}$ corresponds to orthogonal projection onto the subspace $\Span \pi^{(\lambda)}_{j k} \subseteq L^2(G)$.
This correspondence---between projections and convolution against characters---will reappear in different guises throughout this work.
It is easy to see that since $\tilde{f}$ is by definition constant on cosets, its convolution against a character is also constant on cosets, thus $\tilde{f} * \chi^{(\lambda)} \in L^2(G/H)$.
Therefore to find an orthogonal (generalized Fourier) basis on $L^2(G/H)$, it suffices, for every $\lambda$, to find an orthogonal basis of the vector subspace of $\Span \pi^{(\lambda)}_{j k}$ that consists of functions constant on cosets.

Consider a function $f \in \Span \pi^{(\lambda)}_{j k}$.
If we write the condition $f(g \bdot h) = f(g)$ for all ${g \in G, h \in H}$ in the basis $\pi^{(\lambda)}_{j k}(\.)$, where $f$ is expressed as $f(\.) = \sum_{j, k = 1}^{d_{\lambda}} c_{j k} \pi^{(\lambda)}_{j k}(\.)$, we~get
\[
\sum_{j, k = 1}^{d_{\lambda}} \del{\sum_{l=1}^{d_{\lambda}} c_{j l} \pi^{(\lambda)}_{k l}(h)} \pi^{(\lambda)}_{j k}(g) = \sum_{j, k = 1}^{d_{\lambda}} c_{j k} \pi^{(\lambda)}_{j k}(g)
.
\]
Let $\m{C}$ be a matrix with entries $c_{j k}$.
The above condition can be more compactly written as~$\pi^{(\lambda)}(h) \m{C}^{\top} = \m{C}^{\top}$ for all $h \in H$.
Thus, the rows of $\m{C}$ are invariant vectors of the representation~$\pi^{(\lambda)}$ that is restricted onto the subgroup $H$---note that these can be linearly dependent.
Let $r_\lambda$ be the dimension of the space of invariant vectors of $\eval{\pi^{(\lambda)}}_H$, and choose an orthonormal basis $e_1, \ldots, e_{d_{\lambda}}$ of $V_{\lambda}$ such that $e_1, \ldots, e_{r_\lambda}$ are invariant vectors of~$\eval{\pi^{(\lambda)}}_H$.
Then, the rows of $\m{C}$ are linear combinations of $e_1, \ldots, e_{r_\lambda}$ and thus the dimension of the subspace of $\Span \pi^{(\lambda)}_{j k}$ that consists of functions constant on cosets is exactly $r_\lambda d_{\lambda}$.

If $\pi^{(\lambda)}_{j k}$ are defined using the preceding basis of $V_{\lambda}$, a direct computation shows that $\pi^{(\lambda)}_{j k}$, $1 \leq j \leq d_{\lambda}$, $1 \leq k \leq r_\lambda$ are constant on cosets and thus form the orthogonal basis we were originally looking for \cite{weil1940}.

Using these notions, define the \emph{spherical functions} to be the entries of the first $r_\lambda$ columns of the matrix $\pi^{(\lambda)}_{j k}$ for $1 \leq j \leq d_{\lambda}$ and $1 \leq k \leq r_\lambda$, and define the \emph{zonal spherical functions} similarly, but with $1 \leq j \leq r_\lambda$ and $1 \leq k \leq r_\lambda$.
Note that not all spherical functions are zonal, due to the differing indices, and that zonal spherical functions are constant on double cosets, meaning that
\[
\pi^{(\lambda)}_{j k}(h_1 \bdot g \bdot h_2) &= \pi^{(\lambda)}_{j k}(g)
&
h_1, h_2 \in H
&
&1 \leq j \leq r_\lambda
\]
which can be seen by a direct computation.
With these notions at hand, we are now ready to describe stationary Gaussian processes on $G/H$.

\begin{restatable}{theorem}{ThmStationaryHomogeneous} \label{thm:stationary_homogeneous}
A Gaussian process $f \~[GP](0, k)$ on a compact homogeneous space $G/H$ is stationary with respect to the action of $G$ if and only if $k$ is of form
\[ \label{eqn:stationary_homogeneous:kernel}
&
k(g_1 \bdot H, g_2 \bdot H)
=
\sum_{\lambda \in \Lambda}
\sum_{j, k=1}^{r_\lambda}
a^{(\lambda)}_{j k} \Re \pi^{(\lambda)}_{j k}(g_2^{-1} \bdot g_1)
\]
where the coefficients $a^{(\lambda)}_{j k} \in \R$ form symmetric positive semi-definite matrices $\m{A}^{(\lambda)}$ of size $r_\lambda \x r_\lambda$ satisfying $\sum_\lambda \tr \m{A}^{(\lambda)} < \infty$, and $\pi^{(\lambda)}_{j k}$ are the zonal spherical functions.
Moreover, for each individual $\lambda \in \Lambda$, the corresponding sum over $j, k$ is positive semi-definite.
\end{restatable}

\begin{proof}
Similar to before, this follows mostly by specializing Theorem 5 of \textcite{yaglom1961}, which describes complex-valued stationary processes to the real-valued Gaussian case.
A detailed proof is given in \Cref{appdx:proofs}.
\end{proof}

\Cref{thm:stationary_homogeneous} simplifies for \emph{symmetric} homogeneous spaces, which possess certain additional properties that hold, for instance, for all homogeneous spaces of constant sectional curvature.
For such spaces, we study non-compact analogs of this result in part II of this work~\cite{part2}. 
For symmetric spaces, $r_\lambda \leq 1$, thus every representation of $G$ corresponds to at most one zonal spherical function $\pi^{(\lambda)}_{1 1}$.
The matrix $\m{A}^{(\lambda)}$, then, simplifies into a scalar similar to the one which appeared in the Lie group setting.
This occurs, for instance, on spheres $\bb{S}_n = \SO(n+1)/\SO(n)$ and Grassmann manifolds ${\Gr(r, n) = \Ort(n)/(\Ort(r) \x \Ort(n-r))}$, including projective spaces $\RP_n$ as a special case.
We illustrate a Matérn kernel on the sphere, and samples of the corresponding process, in \Cref{fig:sphere-kernel}.

\begin{figure}
\begin{subfigure}{0.2425\textwidth}
\includegraphics[scale=0.25]{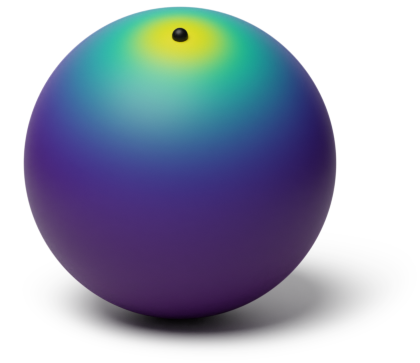}
\caption{Kernel}
\end{subfigure}
\begin{subfigure}{0.2425\textwidth}
\includegraphics[scale=0.25]{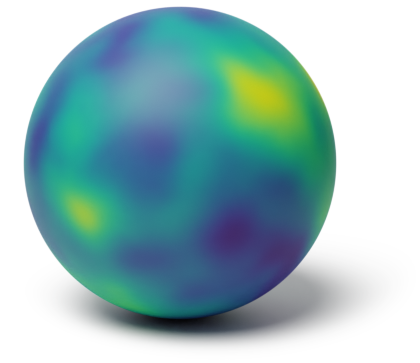}
\caption{Sample}
\end{subfigure}
\begin{subfigure}{0.2425\textwidth}
\includegraphics[scale=0.25]{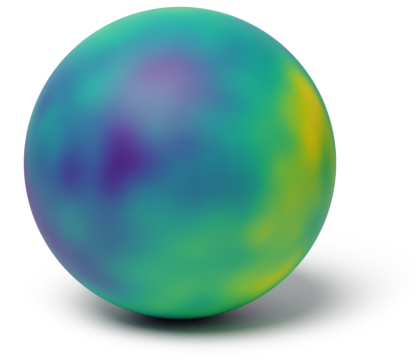}
\caption{Sample}
\end{subfigure}
\begin{subfigure}{0.2425\textwidth}
\includegraphics[scale=0.25]{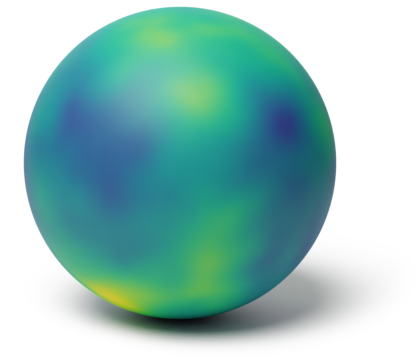}
\caption{Sample}
\end{subfigure}
\caption{Illustration of the Matérn-3/2 kernel of \Cref{sec:heat_matern}, and samples from its corresponding prior Gaussian process, on the sphere $G/H = \SO(3)/\SO(2) = \mathbb{S}_2$.}
\label{fig:sphere-kernel}
\end{figure}

\Cref{thm:stationary_homogeneous} can also be used to give a description of left-invariant kernels on Lie groups, complementing the preceding results on bi-invariant kernels.

\begin{corollary} \label{cor:stationary_left_invar}
Every left-invariant but not necessarily bi-invariant kernel on a compact Lie group $G$ can be expressed according to \Cref{eqn:stationary_homogeneous:kernel}, with $r_\lambda = d_{\lambda}$.
\end{corollary}

\begin{proof}
Letting $H = \cbr{e}$ be the trivial group, the claim is immediate from \Cref{thm:stationary_homogeneous}.
\end{proof}

Note that any stationary positive semi-definite function $k$ on a homogeneous space $G/H$ induces a positive semi-definite function 
\[
&
\Bbbk(g)
=
k(g \bdot H, e \bdot H)
=
k(g_1 \bdot H, g_2 \bdot H)
&
g_2^{-1} \bdot g_1 = g
\]
over the group $G$ which is constant on all double cosets $H \bdot g \bdot H$, namely
\[
\Bbbk(h_1 \bdot g \bdot h_2) = \Bbbk(g)
\]
for $g\in G$ and $h_1,h_2\in H$.
Conversely, any positive semi-definite function on $G$ constant on all double cosets $H \bdot g \bdot H$ induces a positive semi-definite kernel on the homogeneous space~$G/H$.\footnote{Note that $\Bbbk$ should only represent a left-invariant kernel. In particular, this means that $\Bbbk(g_2^{-1} \bdot g_1) = \Bbbk(g_1 \bdot g_2^{-1})$ is \emph{not} necessarily true for all $g_1, g_2 \in G$ and that these kernels do not fall under the assumptions of \Cref{thm:stationary_group} of the preceding section.}

\section{Computational Techniques for Stationary Gaussian Processes} \label{sec:computation}

We now develop computational techniques for implementing stationary Gaussian process models from the preceding section.
To work with the posterior mean, posterior covariance, and pathwise conditioning expressions of \Cref{sec:gp-intro}, we need two computational primitives: 
\1*[(a)] Evaluating the covariance kernel $k(\.,\.)$ pointwise.
\2*[(b)] Efficiently sampling from the prior $f(\.)$.\footnote{For a given set of evaluation locations, $f$ can be sampled by forming and factorizing the kernel matrix at cubic cost with respect to the number of evaluation locations: we use the term \emph{efficient sampling} to refer to sampling $f$ approximately \emph{without these costs}.}
\0*
In general, (a) and (b) need to be performed in an automatically-differentiable manner, for instance to optimize model hyperparameters using maximum likelihood, or differentiate the prior density in Hamiltonian Monte Carlo.

Assuming the scaling coefficients~$a^{(\lambda)}$ or the scaling matrices~$\m{A}^{(\lambda)}$ are a priori fixed, \Cref{thm:stationary_group,thm:stationary_homogeneous} offer avenues for evaluating $k(\.,\.)$ and, with some additional work, approximately sampling $f(\.)$.
This can be achieved by truncating the respective infinite series to the desired accuracy.
To do so, one must be able to traverse the index set $\Lambda$: we study this in \Cref{sec:computation:comp_traverse}.
For pointwise kernel evaluation, one must also compute the characters $\chi^{(\lambda)}$ in the Lie group case, and the zonal spherical functions $\pi^{(\lambda)}_{j, k}$ in the homogeneous space case: 
we study this in \Cref{sec:computation:comp_pointwise}.

For sampling, in \Cref{sec:computation:comp_sampling} we develop a technique we term \emph{generalized random phase Fourier features} that allows efficient sampling using only the aforementioned quantities, as long as one can also draw random samples from the manifold's Haar measure.
Moreover, it allows one---if needed---to compute the Karhunen--Loève basis in terms of the characters or zonal spherical functions, respectively.

With these quantities at hand, all that remains is to obtain the scaling coefficients~$a^{(\lambda)}$ or scaling matrices~$\m{A}^{(\lambda)}$, which are kernel-dependent. 
We study these for specific kernel classes in \Cref{sec:heat_matern}.

\subsection{Enumerating Representations of Compact Lie Groups} \label{sec:computation:comp_traverse}

Both for (approximate) pointwise kernel evaluation and for sampling on compact Lie groups and their homogeneous spaces we will need a way to enumerate the set $\Lambda$ of irreducible unitary representations.
This is a central question of the representation theory.
We now sketch a practical approach to doing so.

Every compact Lie group $G$ contains a connected abelian Lie subgroup, which is isomorphic to a torus---call the maximal (by inclusion) of these subgroups \emph{maximal tori}.
All maximal tori have the same dimension and are conjugate to one another, in the sense that if $T^{(1)}, T^{(2)} \subseteq G$ are two maximal tori, then there exists $g \in G$ such that $T^{(1)} = g^{-1} \bdot T^{(2)} \bdot g$.
Moreover, every element of $G$ is conjugate to an element of a given maximal torus---for example, every unitary matrix which is an element $U \in G = \f{U}(d)$ is diagonalizable by a unitary transformation, namely $U = V D V^{-1}$ where $D$ is an element of the torus of diagonal unitary matrices and $V \in \f{U}(d)$.

Representations of $G$ can be therefore be studied by studying their relationship with representations on the torus, which are simpler.
Specifically, on an $m$-dimensional torus $\bb{T}^m$, the standard theory of Fourier analysis enables one to express a suitably regular function $f : \bb{T}^m \-> \C$ as an infinite sum 
\[
f(x) = \sum_{n \in \Z^m} c_n e^{2\pi i \innerprod{n}{x}} = \sum_{n\in\Lambda} c_n \pi_n(x)
\]
which we have expressed in two ways---analytic, and representation-theoretic.
In the expression above $\Lambda = \Z^m$ and $\pi_n : \bb{T}^m \-> \bb{T} \subseteq \GL(\C^1)$ are the irreducible unitary representations, which on a torus are one-dimensional and can be explicitly written as complex exponentials.
This establishes a bijective correspondence between such representations and Fourier basis functions, at least as long as one studies the torus.

Since every compact Lie group contains a maximal torus $T$, this gives us information about representations on more general compact Lie groups.
When restricted to such a torus, a representation of $G$ splits into a sum of one-dimensional irreducible representations, parametrized by tuples of integers, which are called the \emph{weights} of the representation.
One can show these tuples, and the original representation on $G$, are both uniquely determined by a specific weight called the \emph{highest weight}.
By mapping the set of highest weights with a linear transform, one obtains the set of \emph{signatures}.

\begin{theorem}
The set of irreducible unitary representations of $G$ is in explicit bijective correspondence with the set of \emph{signatures} $\Lambda \subseteq \Z^m$, where $m$ is the dimension of $T$.
\end{theorem}

\begin{proof}
\Cref{appdx:trav_repr}.
\end{proof}

Henceforth, we use $\Lambda$ to refer to the set of signatures, which can be calculated explicitly for most Lie groups.
We provide explicit expressions for $\SO(n)$ and $\SU(n)$ in Appendix~\ref{appdx:weights_examples}.

We can use this bijective correspondence to impose an ordering on $\Lambda$, and define truncated analogs of the infinite series given by \textcite{yaglom1961}.
Specifically, for a stationary kernel $k$ on a compact Lie group $G$, and $L < \infty$, define
\[
k^{(L)}(g_1,g_2) = \sum_{\lambda\in\tl\Lambda} a^{(\lambda)} \Re \chi^{(\lambda)}(g_2^{-1} \bdot g_1)
\]
where $\tl\Lambda\subseteq\Z^m$ is of size $L$ and consists of signatures with the (potentially up to a constant) largest values of $a^{(\lambda)}$, and we take the real part of each character since $k$ is real-valued.
Analogously, on a compact homogeneous space $G/H$ define
\[
k^{(L)}(g_1 \bdot H, g_2 \bdot H) = \sum_{\lambda\in\tl\Lambda} \sum_{j, k=1}^{r_{\lambda}} a^{(\lambda)}_{j k} \Re \pi^{(\lambda)}_{j k}(g_2^{-1} \bdot g_1)
\]
where $\tl\Lambda\subseteq\Z^m$ is of size $L$, and consists of signatures with the (potentially up to a constant) largest values under a suitable matrix ordering, for instance the one given by $\tr\m{A}^{(\lambda)}$.

For the concrete kernels we consider in \Cref{sec:heat_matern}, signatures determine the respective Laplace--Beltrami eigenvalues, and the coefficients $a^{(\lambda)}$ grow monotonically with these eigenvalues.
Similarly, the matrices $\m{A}^{(\lambda)}$ turn out to be scalar multiples of the identity, where the scalar multiple grows monotonically with the respective eigenvalues.
This allows one to obtain a suitable ordering, and therefore implement the truncated series numerically---as long as one can calculate the characters and zonal spherical functions, respectively.
For any truncation level, the resulting expressions are stationary and positive semi-definite.
We therefore proceed to study how to calculate characters and zonal spherical functions.

\subsection{Pointwise Kernel Evaluation} \label{sec:computation:comp_pointwise}

To evaluate a stationary kernel pointwise, we now study how to calculate the characters $\chi^{(\lambda)}$ and zonal spherical functions $\pi^{(\lambda)}_{jk}$ corresponding to a signature $\lambda\in\Lambda$.

\subsubsection{Calculating the Characters of Lie Groups} \label{sec:computation:char}

To calculate the characters of a compact Lie group $G$, we will use the observations that (i) every element of $G$ is conjugate to an element of a fixed maximal torus $T$, (ii) characters are conjugation-invariant, and (iii) on a maximal torus, they admit an explicit description via the Weyl character formula.

\begin{theorem}
Every character $\chi^{(\lambda)}$ with signature $\lambda$ can be expressed as 
\[
\chi^{(\lambda)}(g) &= \chi^{(\lambda)}(t) = \frac{P_1(t)}{P_2(t)}
&
t &= \tilde{g}^{-1} \bdot g \bdot \tilde{g}
\]
where $t \in T$ is an element of a maximal torus conjugate to $g$.
For a fixed identification of~$T$ with $\bb{T}^m \subseteq \C^m$, $P_1$ and $P_2$ are explicit polynomials with respect to the variables $z_1, \ldots z_m, \overline{z_1}, \ldots, \overline{z_m}$ of $\bb{T}^m$.
Moreover, the number of terms in $P_1, P_2$ does not grow with~$\lambda$.
\end{theorem}

\begin{proof}
This is a reformulation of \textcite[Theorem 9.9]{fegan1991}---see \Cref{appdx:character_formula}.
\end{proof}

These expressions are ratios of polynomials, and can therefore be evaluated and automatically differentiated.
To find $t$, one employs techniques that depend on the group.
In many cases, such as $\SU(n)$ or $\SO(2n+1)$, this reduces to finding the eigenvalues of $g$.
Additionally, for many groups, the polynomials $P_1$ and $P_2$ can often be rewritten as certain determinants of order $\dim T$, or short sums thereof.
In cases where $P_2 = 0$, the ratio should be understood as a limit when the denominator approaches zero, for instance when $g$ has repeated eigenvalues. 
This is not an issue, since one can show that the denominator is always a divisor of the numerator, so the ratio $\frac{P_1}{P_2}$ is itself a polynomial.
Further details on these expressions can be found in \Cref{appdx:character_formula}, and explicit expressions for the Lie groups $\SU(n)$ and $\SO(n)$ can be found in \Cref{appdx:particular_character_formulas}.

\subsubsection{Calculating the Zonal Spherical Functions of Homogeneous Spaces} \label{sec:computation:zonal}

\begin{figure}
\begin{subfigure}{0.32\textwidth}
\begin{tikzpicture}
\path[use as bounding box] (-2.625,-2) rectangle (2.125,2);
\node at (0,0) {\includegraphics[scale=0.25]{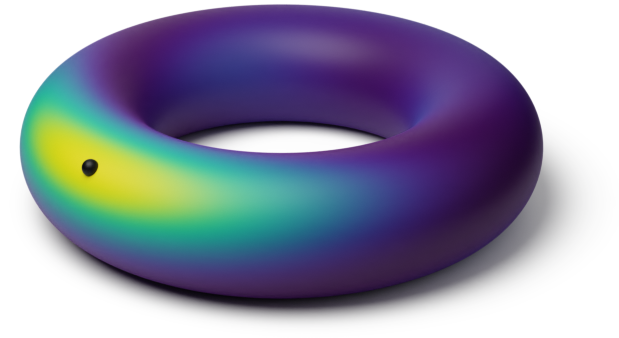}};
\end{tikzpicture}
\caption{Torus}
\end{subfigure}
\begin{subfigure}{0.32\textwidth}
\begin{tikzpicture}
\path[use as bounding box] (-2.625,-2) rectangle (2.125,2);
\node at (0,0) {\includegraphics[scale=0.25]{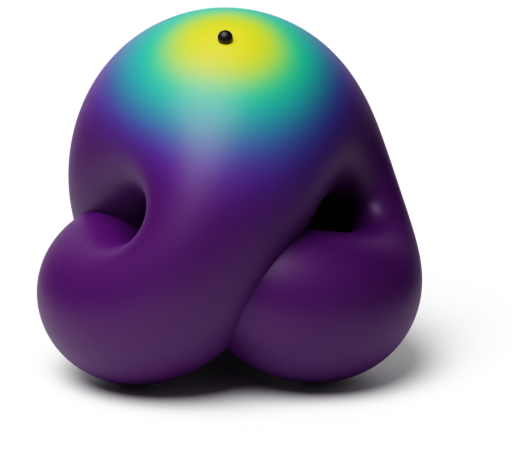}};
\end{tikzpicture}
\caption{Projective plane}
\end{subfigure}
\begin{subfigure}{0.32\textwidth}
\begin{tikzpicture}
\path[use as bounding box] (-2.625,-2) rectangle (2.125,2);
\node at (0,0) {\includegraphics[scale=0.25]{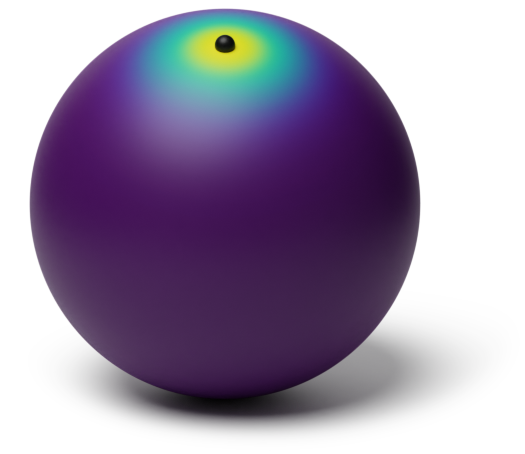}};
\end{tikzpicture}
\caption{Sphere}
\end{subfigure}
\caption{Values of the heat kernel $k(\bullet,\.)$ on $\mathbb{T}^2$, $\RP_2$, and $\mathbb{S}_2$. The torus $\mathbb{T}^2$ is equipped with the flat metric, so the depicted embedding is not isometric.}
\label{fig:gp-kernels}
\end{figure}

The class of homogeneous spaces is very general, and thus a unified treatment similar to the Lie group setting is more difficult.
Consequently, the most efficient ways to compute zonal spherical functions tend to involve space-specific considerations, even in symmetric spaces which contain at most one zonal spherical function per signature.
Nonetheless, we present a general technique to express zonal spherical function on $G/H$ in terms of the characters on $G$, which we term \emph{generalized periodic summation} by analogy with the Euclidean case.
Before presenting this, however, we consider homogeneous spaces where case-specific descriptions of spherical functions are available, at least to some extent.

\begin{itemize}
\item $\bb{S}_n = \SO(n+1)/\SO(n)$: \emph{the $n$-dimensional sphere}.
This is a symmetric homogeneous space.
The spherical functions on this space are called \emph{spherical harmonics}, and are very well-studied.
In particular, zonal spherical harmonics can be explicitly written in terms of the \emph{Gegenbauer~polynomials}.
\item $\f{V}(k, n) = \SO(n)/\SO(n-k)$: \emph{Stiefel manifolds}.
This is the space of all $k$-tuples of orthonormal vectors in $\R^n$.
The spherical functions are called \emph{Stiefel harmonics}, and are studied by \textcite{gelbart1974,strichartz1975}.
\item $\RP_n = \SO(n)/\Ort(n-1)$: \emph{real projective space.}
This is the space of all one-dimensional linear subspaces of $\R^n$, and is a symmetric space.
One can obtain the zonal spherical functions by leveraging connections between this space and the sphere, or by rewriting it as $\f{P}(n) = \Ort(n)/(\Ort(1) \x \Ort(n-1))$, which is as a special case of a Grassmannian.
\item $\Gr(k, n) = \Ort(n)/(\Ort(k) \x \Ort(n-k))$: \emph{Grassmannians.}
This is the space of all $k$-dimensional linear subspaces of $\R^n$, and is a symmetric homogeneous space.
Spherical functions for Grassmannians are described by \textcite{davis1999a, davis1999b}.
\end{itemize}

In cases where space-specific descriptions of spherical functions are available, efficient algorithmic techniques can often be constructed from them.
We illustrate the values of kernels defined in \Cref{sec:heat_matern}, computed by means of such techniques, in \Cref{fig:gp-kernels}.
Unfortunately, many descriptions involve heavy mathematical machinery, and it can be highly non-trivial to translate them into a practical algorithm.
In cases where applications merit it, we view the formulation of such algorithms as a promising avenue for future work, and now return to the general case.

If the coefficients matrices $\m{A}^{(\lambda)} = a^{(\lambda)}\m{I}$ are multiples of the identity matrix, as will be the case for the heat and Matérn kernels studied in \Cref{sec:heat_matern}, one can compute kernels on $G/H$ in terms of characters of $G$.
The next claim shows how.

\begin{restatable}{theorem}{ThmPeriodicSummation} \label{thm:periodic_summation}
Let $f \~[GP](0,k)$, where $k$ is stationary and satisfies $\m{A}^{(\lambda)} = a^{(\lambda)}\m{I}$.
Then
\[
\!
k(g_1 \bdot H ,g_2 \bdot H) &= \int_{H} k_G(g_1 \bdot h, g_2) \d\mu_H(h)
&
f(g \bdot H) &= \int_H f_G(g \bdot h) \d\mu_H(h)
\]
where $\mu_H$ is the Haar measure on $H$, and $f_G \~[GP](0,k_G)$ with
\[
k_G(g_1, g_2)
=
\sum_{\lambda \in \Lambda} a^{(\lambda)} \Re \chi^{(\lambda)}(g_2^{-1} \bdot g_1).
\]
\end{restatable}

\begin{proof}
\Cref{appdx:proofs}.
\end{proof}

This generalizes the \emph{periodic summation} idea in the Euclidean setting, which says that a kernel on the torus can be obtained by starting with a kernel on $\R^d$, and summing over multiples of $2\pi$ \cite{borovitskiy2020}.
One can therefore compute such kernels on a homogeneous space $G/H$ by computing their analogs on the underlying Lie group $G$ using the techniques of \Cref{sec:computation:char}, and averaging over uniform random samples from $H$, obtaining 
\[
k(g_1 \bdot H, g_2 \bdot H) &\approx \frac{1}{S} \sum_{s=1}^S k_G(g_1\bdot h_s, g_2)
&
h_s &\~ \mu_H
.
\]
Note that this procedure---in the finite-sample case---does \emph{not} necessarily preserve positive semi-definiteness (or even symmetry) of the kernel.
To guarantee a well-defined kernel at the cost of additional computational complexity, it can therefore be preferable to average over both arguments of the kernel, obtaining 
\[
k(g_1 \bdot H, g_2 \bdot H) &\approx \frac{1}{S^2} \sum_{s=1}^S \sum_{s'=1}^S k_G(g_1\bdot h_s, g_2\bdot h_{s'})
&
h_s &\~ \mu_H
\]
where we emphasize that a total of $S$ samples are computed from $\mu_H$, which are summed over the set of all pairs of samples.
The resulting expression is guaranteed to be symmetric and positive semi-definite.
We now proceed to consider sampling.

\subsection{Efficient Sampling} \label{sec:computation:comp_sampling}

\begin{figure}
\begin{subfigure}{0.32\textwidth}
\begin{tikzpicture}
\path[use as bounding box] (-2.625,-2) rectangle (2.125,2);
\node at (0,0) {\includegraphics[scale=0.25]{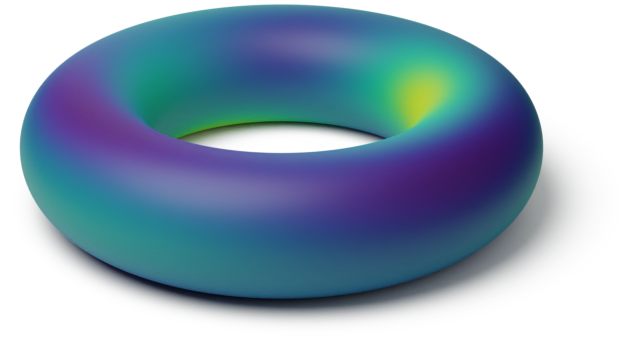}};
\end{tikzpicture}
\caption{Torus}
\end{subfigure}
\begin{subfigure}{0.32\textwidth}
\begin{tikzpicture}
\path[use as bounding box] (-2.625,-2) rectangle (2.125,2);
\node at (0,0) {\includegraphics[scale=0.25]{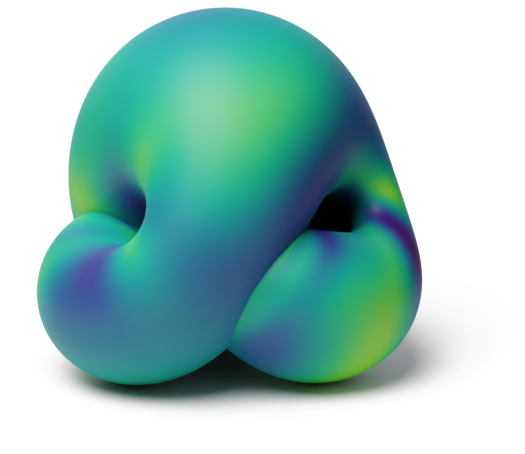}};
\end{tikzpicture}
\caption{Projective plane}
\end{subfigure}
\begin{subfigure}{0.32\textwidth}
\begin{tikzpicture}
\path[use as bounding box] (-2.625,-2) rectangle (2.125,2);
\node at (0,0) {\includegraphics[scale=0.25]{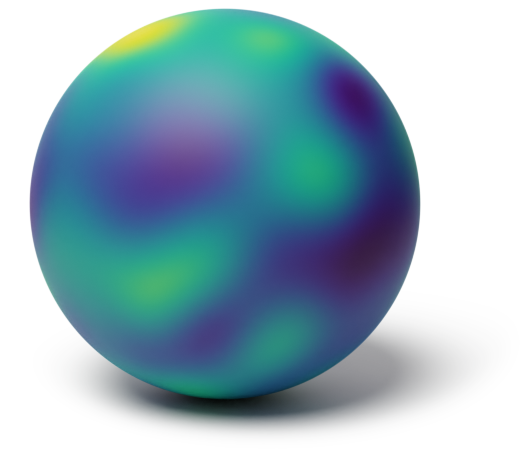}};
\end{tikzpicture}
\caption{Sphere}
\end{subfigure}
\caption{Samples from a Gaussian process with heat kernel covariance, on $\mathbb{T}^2$, $\RP_2$, and~$\mathbb{S}_2$.}
\label{fig:gp-samples}
\end{figure}

Since a sample from a posterior Gaussian process can be obtained by transforming prior samples into posterior samples using \Cref{eqn:intro:pathwise} we now study the problem of sampling from prior Gaussian process $f$.
Given the set of signatures $\Lambda$ of \Cref{sec:computation:comp_traverse}, we can form a prior sample from $f$ by summing a sequence of processes $f^{(\lambda)} \~[GP](0,k^{(\lambda)})$ where $k^{(\lambda)} = a^{(\lambda)} \, \Re\chi^{(\lambda)}(g_2^{-1} \bdot g_1)$ or its homogeneous space analog.
We now introduce techniques for efficiently sampling from such processes, which further allow one to compute a Karhunen--Loève basis as a byproduct.

\subsubsection{Generalized Random Phase Fourier Features} 
\label{sec:computation:comp_sampling:rff}

The key idea for constructing samples comes from the following observation.

\begin{proposition} \label{thm:rkhs}
Let $\c{S}$ be a measurable space, let $\mu$ be probability measure on $\c{S}$, and let $\phi_j$ with $j = 1,\ldots,N$ be an $L^2(\c{S}, \mu)$-orthonormal system.
Let $K(x, x') = \sum_{j=1}^N \phi_j(x) \overline{\phi_j(x')}$.
Then for all $x, x' \in \c{S}$ we have
\[ \label{eqn:rkhs_prop}
K(x, x') = \int_{\c{S}} K(x, u) \overline{K(x', u)} \d  \mu(u)
.
\]
\end{proposition}

\begin{proof}
This follows by direct computation, given for completeness in~\Cref{appdx:proofs}.
\end{proof}

We can therefore Monte Carlo approximate the function $K$ as
\[ \label{eqn:rff_kernel}
K(x, x') &\approx \frac{1}{S}\sum_{s=1}^S K(x, u_s) \overline{K(x', u_s)}
&
u_s &\~ \mu.
\]
If $f \~[GP](0, K)$ then we have the approximation
\[ \label{eqn:rff_process}
f(x) &\approx \frac{1}{\sqrt{S}} \sum_{s=1}^S w_s K(x, u_s)
&
u_s &\~ \mu
&
w_s &\~[N](0, 1)
.
\]
Thus, starting from a covariance kernel of the form $K(x, x') = \sum_{j=1}^N \phi_j(x) \overline{\phi_j(x')}$ for an orthonormal system $\phi_j$, we have constructed a family of random functions $K(\., u_l)$ which enable us to approximately express the Gaussian process $\f{GP}(0, K)$ in a form which is directly amenable to efficient sampling.

If $\c{S} = \bb{T}^d$, and $\phi_j$ are sinusoids from the sine-cosine random Fourier features, then $K(x,u_l)$ are exactly cosines with a random phase shift.
We therefore say that $u_l$ is the \emph{random phase} and $K$ is the \emph{phase function}, and call this construction \emph{generalized random phase Fourier features} for general measure spaces, equipped only with a probability measure $\mu$.
We illustrate samples from such processes, computed additionally using the techniques of \Cref{sec:heat_matern}, in \Cref{fig:gp-samples}.
Note that the process \Cref{eqn:rff_process} may fail to be stationary: we examine this further in \Cref{sec:computation:comp_sampling:kl}.
We now consider spaces of interest.

\subsubsection{Phase Functions and Efficient Sampling for Lie Groups} \label{sec:computation:comp_sampling:lie}

For a compact Lie group $G$, there are usually efficient methods for sampling from the Haar measure on $G$: see \textcite{mezzadri07}. 
Thus, generalized random phase Fourier features enable us to efficiently sample from approximate priors using only the characters.
If $\chi^{(\lambda)}$ is real-valued, namely $\Re \chi^{(\lambda)} = \chi^{(\lambda)}$, then we let $K^{(\lambda)}(g_1, g_2) = d_{\lambda} \chi^{(\lambda)}(g_2^{-1} \bdot g_1)$, corresponding to $\phi_\lambda = \sqrt{d_{\lambda}} \pi^{(\lambda)}_{j k}$.
Otherwise, note that the map $g \|> \del[0]{\pi^{(\lambda)}(g)}^*$, which sends a group element $g$ to the conjugate transpose of the matrix defined by the matrix coefficients $\pi^{(\lambda)}_{ij}(g)$, is an irreducible unitary representation.
It is thus equal to $\pi^{(\lambda')}$ for some $\lambda' \in \Lambda$, $\lambda' \not= \lambda$ and whose character $\chi^{(\lambda')}$ obviously equals $\overline{\chi^{(\lambda)}}$.
In this case, we associate the aggregate of $\sqrt{d_{\lambda}} \pi^{(\lambda)}_{j k}$ and $\sqrt{d_{\lambda'}} \pi^{(\lambda')}_{j k}$ with $\phi_l$, and further associate $d_{\lambda} \chi^{(\lambda)} + d_{\lambda'} \chi^{(\lambda')} = 2 d_{\lambda} \Re \chi^{(\lambda)}$ with $K^{(\lambda)}$.
With this, we can compute each of the terms in~\Cref{eqn:stationary_lie:kernel} using only the characters, and the resulting approximation is always real-valued.
We therefore obtain
\[ 
\label{eqn:random_phase_lie}
f(g)
&\approx
\sum_{\lambda \in \tl{\Lambda}}
\frac{1}{\sqrt{S}}
\sum_{s = 1}^{S}
w^{(\lambda)}_{s}
K^{(\lambda)}(g, u_s)
&
u_s &\~ \mu_G
&
w^{(\lambda)}_{s} &\~[N]\del{0, \frac{a^{(\lambda)}}{d_{\lambda}}}
.
\]
Here, $\tl\Lambda$ are the truncated signatures defined previously, $\mu_G$ is the probabilistic Haar measure on the group $G$.
$\tl{\Lambda}$ should be chosen so as \emph{not} to contain~$\lambda'$ together with any $\lambda$ except in the case $\lambda = \lambda'$.
The size of~$\tl\Lambda$ and value of $L$ control approximation quality.
One can interpret the function $g\|> [K^{(\lambda)}(g, u_s)]_{s=1,\ldots,S,\,\lambda\in\tl\Lambda}$ as an approximate finite-dimensional feature map: one can therefore also use this construction for other purposes, beyond just sampling, where feature maps are helpful.
This can include, for instance, speeding up training using Bayesian-linear-model-based techniques.

\subsubsection{Phase Functions and Efficient Sampling for Homogeneous Spaces} \label{sec:computation:comp_sampling:homogeneous}

If $G/H$ is a homogeneous space, the key question for defining generalized random phase Fourier features is how to correctly choose the function $K$ in \Cref{thm:rkhs}.
A natural idea is to associate each of zonal spherical function $\pi^{(\lambda)}_{j k}(g_2^{-1} \bdot g_1)$ with a separate function $\tl{K}(g_1 \bdot H, g_2 \bdot H)$: however, this is not the right thing to do, as we obtain 
\[
\tl{K}(g_1 \bdot H, g_2 \bdot H) = \sum_{j=1}^{d_{\lambda}} \pi^{(\lambda)}_{j k_1}(g_1) \overline{\pi^{(\lambda)}_{j k_2}(g_2)}
\]
for $1 \leq k_1, k_2 \leq r_\lambda$, which is not the form required by \Cref{thm:rkhs} unless~${k_1 = k_2}$.
Moreover, $\tl{K}$ and the induced approximation $f$ can be complex-valued.
A more delicate approach is therefore needed.
Let $a^{(\lambda)}_{k_1 k_2}$ be a set of coefficients such that they form a symmetric positive semi-definite matrix whose precise form is to be determined later. 
Define 
\[ \label{eqn:homogeneous_K_defn}
K(g_1 \bdot H, g_2 \bdot H)
=
\sum_{k_1, k_2=1}^{r_\lambda}
a^{(\lambda)}_{k_1 k_2} \pi^{(\lambda)}_{k_1 k_2}(g_2^{-1} \bdot g_1)
=
\sum_{k_1, k_2=1}^{r_\lambda}
a^{(\lambda)}_{k_1 k_2}
\sum_{j=1}^{d_{\lambda}}
\pi^{(\lambda)}_{j k_1}(g_1)
\overline{\pi^{(\lambda)}_{j k_2}(g_2)}
\]
and observe the following.

\begin{restatable}{proposition}{ThmKPhaseCond} \label{thm:K_phase_cond}
The function $K$ defined by~\Cref{eqn:homogeneous_K_defn} satisfies
\[ \label{eqn:rkhs_prop_to_test}
K(x, x') = \int_{G/H} K(x, u) \overline{K(x', u)} \d\mu_{G/H}(u)
\]
if and only if the matrix $\m{A} = d_{\lambda}^{-1} \del[1]{a^{(\lambda)}_{k_1 k_2}}$ is idempotent, namely ${\m{A}^2=\m{A}}$.
\end{restatable}
\begin{proof}
\Cref{appdx:proofs}.
\end{proof}

It follows that if the scaled coefficient matrix $\m{A}$ whose entries are $d_{\lambda}^{-1} a^{(\lambda)}_{k_1 k_2}$ is an orthogonal projection, then generalized random phase Fourier features sampling is possible.
By positive semi-definiteness, $\m{A}$ has an eigendecomposition
\[
\m{A} &= \m{U}_{(\lambda)} \operatorname{diag}\del{\alpha_1^{(\lambda)}, \ldots, \alpha_{r_\lambda}^{(\lambda)}}\m{U}_{(\lambda)}^{\top}
&
\m{U}_{(\lambda)} &= \del{\v{u}_1^{(\lambda)}, \ldots, \v{u}_{r_\lambda}^{(\lambda)}}
\]
where $\alpha_k^{(\lambda)}$ are non-negative eigenvalues, and the columns of $\m{U}_{(\lambda)}$ are an orthonormal basis of eigenvectors of $\m{A}$.
Reinterpreting this matrix product as a sum, we can write
\[
\m{A} = \sum_{k=1}^{r_\lambda} \alpha_k^{(\lambda)} \ubr{\v{u}_k^{(\lambda)} {\v{u}_k^{(\lambda)}}^{\top}}_{\m{U}^{(k)}_{(\lambda)}}
\]
where $\m{U}^{(k)}_{(\lambda)}$ are orthogonal projection matrices.
Now, define
\[
\tl{K}^{(\lambda)}_k(g_1 \bdot H, g_2 \bdot H) = \sum_{k_1, k_2=1}^{r_\lambda}
d_{\lambda} \m{U}^{(k)}_{\lambda, k_1 k_2} \pi^{(\lambda)}_{k_1 k_2}(g_2^{-1} \bdot g_1).
\]
Recall that the aforementioned map $g \|> \del[0]{\pi^{(\lambda)}(g)}^*$ defines an irreducible unitary representation and thus equal to $\pi^{(\lambda')}$ for some $\lambda' \in \Lambda$.
If $\lambda$ is such that $\lambda = \lambda'$, then $\pi^{(\lambda)}_{j k} = \overline{\pi^{(\lambda)}_{k j}}$ which, by virtue of $\m{U}^{(k)}_{(\lambda)}$ being symmetric, implies $\tl{K}^{(\lambda)}_k (x, x') \in \R$.
We have in this case
\[
\sum_{k_1, k_2=1}^{r_\lambda}
a^{(\lambda)}_{k_1 k_2} \Re \pi^{(\lambda)}_{k_1 k_2}(g_2^{-1} \bdot g_1)
=
\sum_{k=1}^{r_\lambda} \frac{\alpha_k^{(\lambda)}}{d_{\lambda}} \tl{K}^{(\lambda)}_k (g_1 \bdot H, g_2 \bdot H)
\]
where $\tl{K}^{(\lambda)}_k$ are real-valued phase functions which for $\lambda = \lambda'$ we rename to $K^{(\lambda)}_k$.

If $\lambda \not= \lambda'$, to ensure the resulting samples are real-valued, write
\[
&\sum_{k_1, k_2=1}^{r_\lambda}
a^{(\lambda)}_{k_1 k_2} \Re \pi^{(\lambda)}_{k_1 k_2}(g_2^{-1} \bdot g_1)
+
\sum_{k_1, k_2=1}^{r_\lambda}
a^{(\lambda')}_{k_1 k_2} \Re \pi^{(\lambda')}_{k_1 k_2}(g_2^{-1} \bdot g_1)
\\
&=\!\!\!\! \sum_{k_1, k_2=1}^{r_\lambda}
a^{(\lambda)}_{k_1 k_2} \pi^{(\lambda)}_{k_1 k_2}(g_2^{-1} \bdot g_1)
+
a^{(\lambda')}_{k_1 k_2} \pi^{(\lambda')}_{k_1 k_2}(g_2^{-1} \bdot g_1)
=
\sum_{k=1}^{r_\lambda} \frac{\alpha_k^{(\lambda)}}{d_{\lambda}} 2 \Re \tl{K}^{(\lambda)}_k (g_1 \bdot H, g_2 \bdot H)
\]
where we used the fact that $a^{(\lambda)}_{k_1 k_2} = a^{(\lambda')}_{k_2 k_1} \in \R$.
From the preceding argument it is clear that $2 \Re \tl{K}^{(\lambda)}_k$ are real-valued phase functions which for $\lambda \not= \lambda'$ we rename to $K^{(\lambda)}_k$.

Since terms with different $\lambda$ values in the generalized random phase Fourier feature approximation are independent, it suffices to sample from each $K^{(\lambda)}_k$ by \Cref{eqn:rff_process}, multiply by $\sqrt{\alpha_k^{(\lambda)} / d_{\lambda}}$, and reassemble the expression by summing terms.
This gives the approximation
\[ 
\label{eqn:random_phase_homogeneous}
f(g \bdot H)
&\approx
\sum_{\lambda \in \tl\Lambda}
\sum_{k = 1}^{r_\lambda}
\frac{1}{\sqrt{S}}
\sum_{s = 1}^{S}
w^{(\lambda)}_{s}
K^{(\lambda)}_k (g \bdot H, u_s)
&
w^{(\lambda)}_{s} &\~[N]\del{0, \frac{\alpha^{(\lambda)}_k}{d_{\lambda}}}
&
u_s &\~ \mu_{G/H}
\]
where $\tl\Lambda$ are the truncated signatures, and~$\mu_{G/H}$ is the probabilistic Haar measure on~$G/H$.
$\tl\Lambda$ should be chosen so as \emph{not} to contain~$\lambda'$ together with any $\lambda$ unless $\lambda = \lambda'$, where this notation was previously defined in the context of the matrix-coefficient conjugate-transpose map $g \|> (\pi^{(\lambda)}(g))^*$.
The size of~$\tl\Lambda$ and the value of $S$ control the quality of the approximation.
If~$G/H$ is a symmetric space, then~$r_\lambda \leq 1$, and the above formulas simplify considerably.
Mirroring the Lie group setting, the function $x \|> [K^{(\lambda)}_k(x, u_s)]_{s=1,\ldots,S,\,k=1,\ldots,r_\lambda,\,\lambda\in\tl\Lambda}$ gives an approximate finite-dimensional feature map for the kernel.

\subsubsection{Calculating a Karhunen--Loève Basis} \label{sec:computation:comp_sampling:kl}

It is not hard to see that matrix coefficients or general spherical functions provide a Karhunen--Loève basis for the Gaussian process, whose covariance is given by the respective character or a combination of zonal spherical functions.
However, the former are available on a much smaller class of spaces compared to the latter.
This was one of the main reasons we explored how to sample from a stationary Gaussian process using only the former.
We now ask and answer a complementary question: given only the characters or zonal spherical functions, is it possible to obtain a Karhunen--Loève basis?

Our strategy will be to find a non-orthonormal basis, and Gram--Schmidt orthogonalize it to obtain a Karhunen--Loève basis.
This generalizes the approach of \textcite[Section 1.3]{dai2013} from spheres to more general spaces.
We begin from an abstract formulation.
Consider an orthonormal system $\phi_j$ in the space $L^2(\c{S}, \mu)$, where $(\c{S}, \mu)$ is a measure space and $\mu$ is a probability measure.
Define, as before, the function $K(x, x') = \sum_{j = 1}^{N} \phi_j(x) \overline{\phi_j(x')}$.
Consider a set of $N$ points $x_j$ in $\c{S}$ and define the matrices
\[
\m{M}_k
=
\begin{pmatrix}
\phi_1(x_1) & \cdots & \phi_1(x_k) \\
\vdots & \ddots & \vdots \\
\phi_k(x_1) & \cdots & \phi_k(x_k) \\
\end{pmatrix}
.
\]

\begin{restatable}{lemma}{ThmFundamentalSetExists} \label{thm:fundamental_set_exists}
There exist $\cbr{x_1, \ldots, x_N} \subseteq \c{S}$ such that $\det \m{M}_k \neq 0$ for all $1 \leq k \leq N$.
\end{restatable}

\begin{proof}
\Cref{appdx:proofs}.
\end{proof}

\begin{definition}
A collection $\cbr{x_j}_{j=1}^N \subseteq \c{S}$ is a \emph{fundamental set} if $\det \m{M}_N \neq 0$.
\end{definition}

The existence of such sets, which are non-unique, is guaranteed by Lemma~\ref{thm:fundamental_set_exists}.

\begin{proposition} \label{thm:fundamental_basis_independent}
The notion of a fundamental set does not depend on a specific orthonormal system system $\phi_j$ but rather only on the space $\Span \phi_j$---that is, on $K$.
\end{proposition}

\begin{proof}
We have
\[ \label{eqn:mnsq}
\m{M}_N^{*} \m{M}_N
=
\begin{pmatrix}
K(x_1, x_1) & \cdots & K(x_N, x_1) \\
\vdots & \ddots & \vdots \\
K(x_1, x_N) & \cdots & K(x_N, x_N)
\end{pmatrix}
.
\]
We have $\det \m{M}_N \neq 0$ if and only if $\det \m{M}_N^{*} \m{M}_N > 0$.
The latter depends on $K$ only, and hence only on $\Span \phi_j$.
\end{proof}

\begin{restatable}{theorem}{ThmFundamentalSpan} \label{thm:fundamental_span}
If $x_j$ is a fundamental set with respect to the system $\phi_j$, then
\[
\Span \phi_j
=
\Span K(\., x_j)
\]
\end{restatable}

\begin{proof}
\Cref{appdx:proofs}.
\end{proof}

Notice that \Cref{eqn:mnsq} allows one to check whether a set of points is a fundamental set without using the basis $\phi_j$.
Thus, one can find a fundamental set by uniformly sampling a set of points $x_1, \ldots, x_N \in G$, and checking whether the determinant $\det \m{M}_N^{*} \m{M}_N$ vanishes.
This is generally a rare event, so in practice it typically suffices to sample once.

The basis $K(x_j,\.)$ is not necessarily orthonormal, but pairwise inner products of basis functions can nonetheless be easily calculated using its special form, by writing
\[
\innerprod{K(\., x_j)}{K(\., x_k)}_{L^2(\c{S}, \mu)}
=
\sum_{n=1}^N \overline{\phi_n(x_j)} \phi_n(x_k)
=
K(x_k, x_j)
=
\overline{K(x_j, x_k)}.
\]
Here, we have used \Cref{eqn:rkhs_prop} and $K(x, x') = \overline{K(x', x)}$.
One therefore obtains a Karhunen--Loève basis by Gram--Schmidt orthogonalizing $K(\cdot, x_j)$ at $\c{O}(N^3)$ cost.

\section{Computational Techniques for Heat and Matérn Kernels} \label{sec:heat_matern}

In the preceding sections, we obtained Lie-theoretic and representation-theoretic expressions for stationary kernels and Gaussian processes under group action.
In these expressions, only the scaling factors $a^{(\lambda)}$ and $\m{A}^{(\lambda)}$, respectively, depend on the specific kernel.
We therefore study how to compute them for the \emph{Matérn} and \emph{heat} kernel classes: our approach will be to define heat kernels first, and extend to the Matérn class using connections between them.
To proceed, we need to turn our spaces into Riemannian manifolds by defining a metric tensor, for which we make the canonical choice.

\begin{Assumption}
$X$ is equipped with the metric tensor induced by its Killing form.
\end{Assumption}

If $X=G$ is a Lie group, this metric is automatically bi-invariant under the action of $G$ on itself from both sides. 
If $X = G/H$ is a homogeneous space, this metric is the quotient metric obtained from $G$, and is invariant under the canonical action of $G$ on~$G/H$.

\subsection{Squared Exponential Kernels as Solutions of the Heat Equation} \label{sec:heat_matern:heat}

Before defining the general Riemannian heat kernel, consider first its Euclidean analog.
In this setting, it is widely known as the squared exponential, or Gaussian, or RBF kernel.
We denote it by $k_{\infty,\kappa,\sigma^2}$.
It has two positive numerical parameters $\kappa$ and $\sigma^2$, representing its amplitude (variance) and length scale, respectively.
For $\v{x}, \v{x}' \in \R^n$ this kernel is
\[
k_{\infty, \kappa, \sigma^2}(\v{x}, \v{x}')
=
\sigma^2
e^{- \frac{\norm[0]{\v{x}-\v{x}'}^2}{2 \kappa^2}}
.
\]
Observe that if we let $\kappa = \sqrt{2 t}$ and $\sigma^2 = \frac{1}{(4 \pi t)^{n/2}}$ then, defining 
\[
\c{P}(t,\v{x},\v{x}') = k_{\infty, \sqrt{2 t}, (4 \pi t)^{-n/2}}(\v{x}, \v{x}')
\]
we can express this kernel as the fundamental solution of the heat equation
\[ \label{eqn:heat_defn}
\frac{\partial \c{P}}{\partial t}(t, \v{x}, \v{x}')
&=
\Delta_{\v{x}} \c{P}(t, \v{x}, \v{x}')
&
\c{P}(0,\v{x},\v{x}') &= \delta(\v{x} - \v{x}')
\]
where $\Delta_{\v{x}}$ is the Laplacian acting on the $\v{x}$-argument, and the initial condition $\delta(\v{x} - \v{x}') \in \c{S}'(X)$ is the Dirac delta function, understood as a vector in the Schwartz space~$\c{S}'(X)$ of tempered distributions.
This gives the general way of defining general analogs of squared exponential kernels, which we call heat kernels, and will be our starting point.

The heat equation generalizes to a Riemannian manifold $X$ under mild regularity conditions, so long as $\Delta$ is replaced with the Laplace--Beltrami operator.
The Riemannian heat kernel $\c{P}$ can now be defined as the fundamental solution of the Riemannian heat equation. We also define its squared-exponential-like reparametrization by
\[ \label{eqn:sq_exp_heat}
&
k_{\infty, \kappa, \sigma^2}(x, x')
=
\frac{\sigma^2}{C_{\kappa}} \c{P}(\kappa^2/2, x, x')
\]
where $C_{\kappa}$ is a constant chosen depending on how one wishes to normalize the kernel's amplitude.
For spaces where $\c{P}(\kappa^2/2, x, x)$ does not depend on $x$, such as homogeneous spaces, a natural choice is $C_{\kappa} = \c{P}(\kappa^2/2, x, x)$.
For other spaces, if $X$ is compact, one can choose $C_{\kappa} = \int_X \c{P}(\kappa^2/2, x, x) \d\!\vol(\v{x})$.
By general theory \cite{grigoryan2009}, these functions are positive definite kernels.
At this stage, however, it is not clear whether or not they are stationary, and if so, what the corresponding scaling coefficients $a^{(\lambda)}$ or scaling matrices $\m{A}^{(\lambda)}$ are.
To understand this, we specialize to the settings at hand.

\subsubsection{Compact Lie Groups}

To understand how to express solutions of the heat equation in a more explicit manner, we begin by connecting the representation-theoretic notions introduced previously with spectral properties of the Laplace--Beltrami operator~$\Delta$.

\begin{restatable}{result}{ThmLaplaceEigenfunctionsGroup} \label{thm:laplace_eigenfunctions_group}
For every $\lambda$, the matrix coefficients $\pi^{(\lambda)}_{jl}$ are eigenfunctions of $-\Delta$ corresponding to the same eigenvalue $\alpha_\lambda \geq 0$.
Moreover, the eigenvalues $\alpha_\lambda$ can be written\footnote{This expression is sometimes called \emph{Freudenthal's formula for the
eigenvalues}.}
\[ \label{eqn:laplace_eigenvalues_lie}
\alpha_\lambda =  \norm{\rho}^2_B - \norm{w+\rho}^2_B
\]
where $w$ is the highest weight of the representation with signature $\lambda$, $\norm{g}^2_B = -B(g, g)$, $B$ is the Killing form of $G$, and $\rho$ is the half-sum of positive roots of $G$---see \Cref{appdx:trav_repr}.
\end{restatable}

\begin{proof}
The first part of the claim is well-known: we sketch a simple but instructive proof in~\Cref{appdx:proofs}.
The second part is given in \textcite[Theorem~10.6]{fegan1991}.
\end{proof}

For technical details on these quantities, see \Cref{appdx:trav_repr}.
With this connection established, one can calculate an explicit expression for the heat kernel on $G$ in the sense of the heat equation using the representation-theoretic notions at hand.

\begin{proposition}
The heat kernel is stationary, and is given by
\[ \label{eqn:heat_lie}
k_{\infty, \sqrt{2 t}, (4 \pi t)^{-n/2}}(g_1, g_2) = \c{P}(t, g_1, g_2)
=
\sum_{\lambda \in \Lambda}
e^{-\alpha_{\lambda} t} d_{\lambda} \chi_{\lambda}(g_2^{-1} \bdot g_1).
\]
\end{proposition}

\begin{proof}
Write out both sides of the heat equation in the basis $\cbr[1]{\sqrt{d_{\lambda}}\pi^{(\lambda)}_{j k}}_{\lambda, j, k}$.
\end{proof}

Therefore, to evaluate heat kernels pointwise, and to sample from the corresponding Gaussian processes, it suffices to apply the machinery of \Cref{sec:computation}.
In particular, note that $\lambda \|> e^{-\alpha_\lambda t}$ is monotone---we will see this holds for all concrete kernels that we consider. 
Therefore, to choose the set of truncated signatures $\tl\Lambda$, we select the signatures corresponding to the smallest eigenvalues, and apply \Cref{thm:laplace_eigenfunctions_group}. 

\subsubsection{Homogeneous Spaces of Compact Lie Groups}

\begin{figure}
\begin{subfigure}{0.32\textwidth}
\includegraphics[scale=0.25]{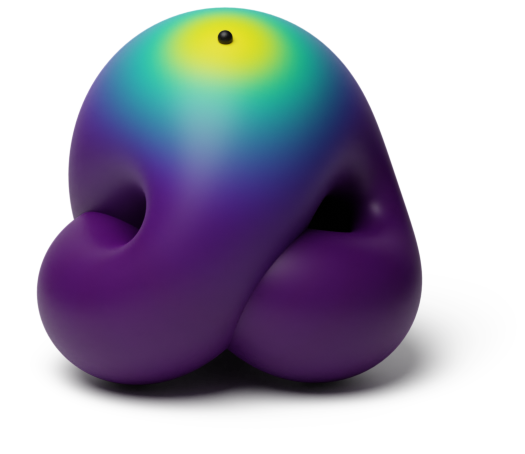}
\caption{$\kappa = 0.1$}
\end{subfigure}
\begin{subfigure}{0.32\textwidth}
\includegraphics[scale=0.25]{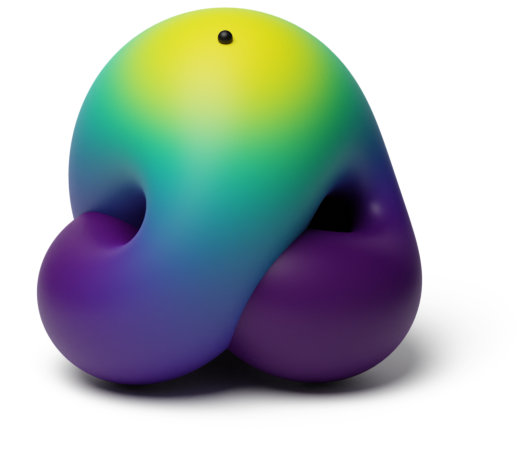}
\caption{$\kappa = 0.25$}
\end{subfigure}
\begin{subfigure}{0.32\textwidth}
\includegraphics[scale=0.25]{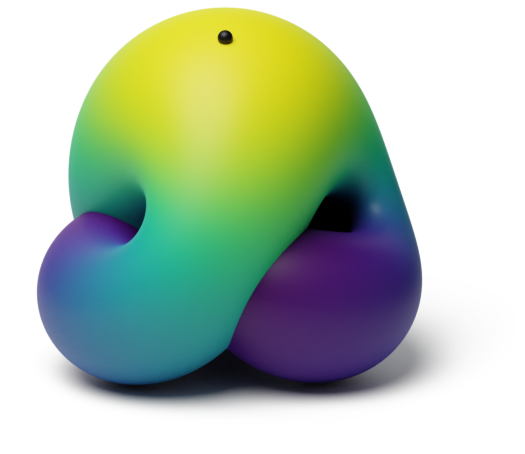}
\caption{$\kappa = 0.5$}
\end{subfigure}
\caption{Values of heat kernels on a projective plane $\RP_2$, with various length scales.}
\label{fig:heat-kernel}
\end{figure}

We again begin by connecting representation-theoretic notions with spectral theory of $\Delta$.

\begin{restatable}{result}{ThmLaplaceEigenfunctionsHomogeneous} \label{thm:laplace_eigenfunctions_hom}
The spherical functions $\pi^{(\lambda)}_{jk}$ on a compact homogeneous space $G/H$ are the eigenfunctions of $-\Delta$.
Moreover, the eigenvalue corresponding to every spherical function~$\pi^{(\lambda)}_{jk}$, which by definition is also a matrix coefficient of $G$, is given in \Cref{eqn:laplace_eigenvalues_lie}.
\end{restatable}

\begin{proof}
\Cref{appdx:proofs}.
\end{proof}

We can use this to compute explicit expressions for the heat kernel on~$G/H$.

\begin{proposition}
The heat kernel is stationary, and is given by
\[ \label{eqn:heat_homogeneous}
k_{\infty, \sqrt{2 t}, (4 \pi t)^{-n/2}}(g_1 \bdot H, g_2 \bdot H) = \c{P}(t, g_1 \bdot H, g_2 \bdot H)
=
\sum_{\lambda \in \Lambda}
e^{-\alpha_{\lambda} t} d_{\lambda}
\sum_{j=1}^{r_\lambda}
\pi^{(\lambda)}_{j j}(g_2^{-1} \bdot g_1).
\]
\end{proposition}

\begin{proof}
Write out both sides of the heat equation in the basis of spherical functions.
\end{proof}

In homogeneous spaces, too, it therefore suffices to apply the machinery of \Cref{sec:computation}.
We illustrate heat kernels computed in this manner in \Cref{fig:heat-kernel}.

\subsection{Matérn Kernels} \label{sec:heat_matern:matern}

We now use the preceding notions to introduce Matérn kernels.
On~$\R^n$, these constitute \cite{rasmussen2006,stein1999} a widely used three parameter family
\[
k_{\nu, \kappa, \sigma^2}(\v{x},\v{x}')
=
\sigma^2
\frac{2^{1-\nu}}{\Gamma(\nu)}
\del{
  \sqrt{2\nu}
  \frac{\norm{\v{x}-\v{x}'}}{\kappa}
}^\nu
K_\nu\del{
  \sqrt{2\nu}
  \frac{\norm{\v{x}-\v{x}'}}{\kappa}
}
\]
with parameters $\sigma^2 > 0$, $\kappa > 0$, and $\nu > 0$ being the variance, length scale, and smoothness, respectively.
Here, $K_\nu$ is the modified Bessel function of the second kind \cite{gradshteyn2014}.
We are interested in defining them in settings beyond $\R^n$.

We approach the problem by connecting Matérn kernels with squared exponential kernels, and leveraging the preceding results.
Still on $\R^n$, as $\nu\->\infty$, we have 
\[
\lim_{\nu\->\infty} k_{\nu, \kappa, \sigma^2}(\v{x},\v{x}') = k_{\infty, \kappa, \sigma^2}(\v{x},\v{x}')
\]
 but this is not the only connection between these families of kernels.
One can show
\[ \label{eqn:matern_integral_formula_1}
k_{\nu, \kappa, \sigma^2}(\v{x}, \v{x}')
&=
\frac{(2\nu)^{\nu}}{\Gamma(\nu)\kappa^{2\nu}}
\int_0^{\infty}
u^{\nu - 1}
e^{-\frac{2 \nu}{\kappa^2} u}
k_{\infty, \sqrt{2 u}, \sigma^2}(\v{x}, \v{x}')
\d u
\\ \label{eqn:matern_integral_formula_2}
&=
\sigma^2 \frac{(2\nu)^{\nu} (4 \pi)^{n/2}}{\Gamma(\nu)\kappa^{2\nu}}
\int_0^{\infty}
u^{\nu - 1 + n/2}
e^{-\frac{2 \nu}{\kappa^2} u}
\c{P}(u, \v{x}, \v{x}')
\d u
.
\]
This can be understood by observing that the spectral measure of the Matérn kernel, which is a $\f{T}$ distribution, can be written as a gamma mixture of Gaussians, which are the spectral measure of the squared exponential kernel.
\Cref{eqn:matern_integral_formula_1,eqn:matern_integral_formula_2} are the kernel analog of this: \Cref{appdx:proofs} gives a proof.
More generally, we adopt this expression---after dropping inessential constants and renormalizing---as the \emph{definition} of a Matérn kernel.

\begin{definition} \label{dfn:matern}
Let $X$ be a $n$-dimensional compact Lie group or a homogeneous space of a compact Lie group, and let $\c{P}(u, x, x')$ be the heat kernel.
Define the \emph{Matérn kernel}~by
\[ \label{eqn:matern_dfn}
k_{\nu, \kappa, \sigma^2}(x, x')
&=
\frac{\sigma^2}{C_{\nu, \kappa}}
\int_0^{\infty}
u^{\nu - 1 + n/2}
e^{-\frac{2 \nu}{\kappa^2} u}
\c{P}(u, x, x')
\d u
\]
where $C_{\nu, \kappa}$ is chosen so that $k_{\nu, \kappa, \sigma^2}(x, x) = \sigma^2$.
\end{definition}

It is easy to check that such kernels are positive (semi)definite whenever the respective squared exponential kernels are positive (semi)definite, regardless of the underlying space~$X$---see \Cref{appdx:proofs}.
The main issue at this stage is that working with Matérn kernels ostensibly requires one to evaluate an integral.
We therefore study the expression for the spaces at hand to understand simplifications, given as follows.

\begin{figure}
\begin{subfigure}{0.32\textwidth}
\includegraphics[scale=0.25]{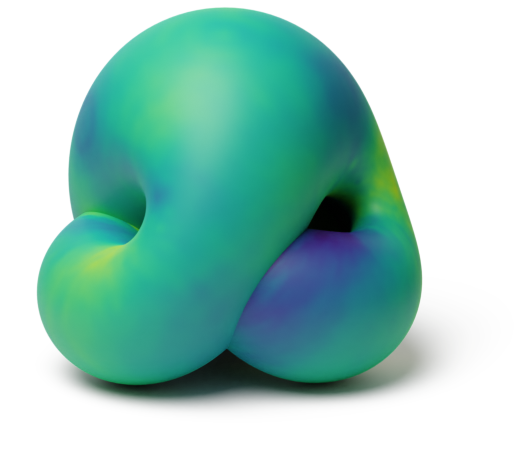}
\caption{$\nu =  0.5$}
\end{subfigure}
\begin{subfigure}{0.32\textwidth}
\includegraphics[scale=0.25]{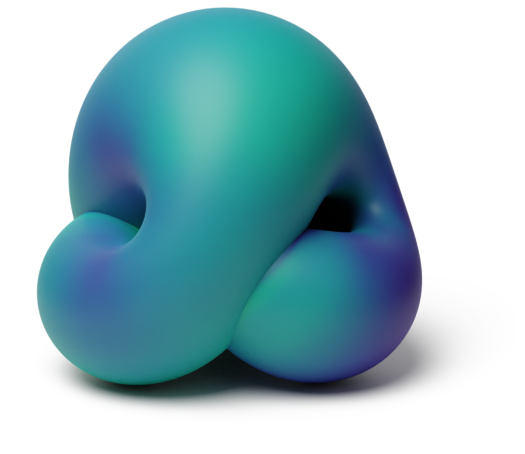}
\caption{$\nu = 1.5$}
\end{subfigure}
\begin{subfigure}{0.32\textwidth}
\includegraphics[scale=0.25]{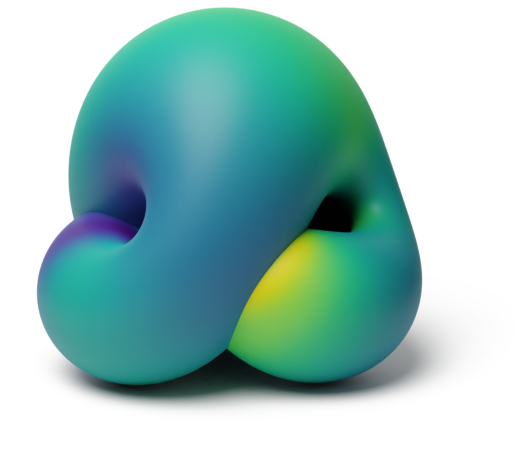}
\caption{$\nu = 2.5$}
\end{subfigure}
\caption{Samples from Matérn Gaussian processes on $\RP_2$, with varying smoothness.}
\label{fig:matern-samples}
\end{figure}

\begin{restatable}{proposition}{ThmMaternLie} \label{thm:matern_lie}
The Matérn kernel on a compact Lie group $G$ is given by
\[ \label{eqn:matern_lie}
k_{\nu, \kappa, \sigma^2}(g_1, g_2)
=
\frac{\sigma^2}{C_{\nu, \kappa}'}
\sum_{\lambda \in \Lambda}
\del{\frac{2 \nu}{\kappa^2} + \alpha_{\lambda}}^{-\nu-n/2}
\,\,
d_{\lambda}
\chi_{\lambda}(g_2^{-1} \bdot g_1)
\]
where $C_{\nu, \kappa}'$ is a normalizing constant which ensures that $k_{\nu, \kappa, \sigma^2}(g, g) = \sigma^2$.
\end{restatable}
\begin{proof}
\Cref{appdx:proofs}.
\end{proof}

\begin{proposition}
The Matérn kernel on a homogeneous space $G/H$ is given by
\[ \label{eqn:matern_homogeneous}
k_{\nu, \kappa, \sigma^2}(g_1 \bdot H, g_2 \bdot H)
=
\frac{\sigma^2}{C_{\nu, \kappa}'}
\sum_{\lambda \in \Lambda}
\sum_{j=1}^{r_\lambda}
\del{\frac{2 \nu}{\kappa^2} + a_{\lambda}}^{-\nu-n/2}
\,\,
d_{\lambda}
\pi^{(\lambda)}_{j j}(g_2^{-1} \bdot g_1).
\]
where $C_{\nu, \kappa}'$ is a normalizing constant which ensures that $k_{\nu, \kappa, \sigma^2}(x, x) = \sigma^2$.
\end{proposition}
\begin{proof}
\Cref{appdx:proofs}, via the same proof as for \Cref{thm:matern_lie}.
\end{proof}

Note that the constants $C_{\nu, \kappa}'$ in~\Cref{eqn:matern_homogeneous,eqn:matern_lie} can differ, owing to the different settings.
All these kernels can be computed using the techniques of \Cref{sec:computation}.
One can similarly derive formulas for random sampling from Matérn priors.
We illustrate such samples in \Cref{fig:matern-samples}.
We conclude by verifying that these kernels satisfy the same differentiability properties as the usual Euclidean Matérn class.

\begin{restatable}{theorem}{ThmSmoothnessComp} \label{thm:smoothness_comp}
For all $\nu > 0$, Matérn kernels on compact Lie groups and their homogeneous spaces are positive definite and lie in the Sobolev space $H^{\nu+n/2}$. 
Thus, they are continuous and possess continuous derivatives of all orders strictly less than $\nu$.
\end{restatable}
\begin{proof}
\Cref{appdx:proofs}.
\end{proof}

\subsection{Approximation Error and Computational Costs}
\label{sec:approx}

To compute the preceding kernels or efficiently sample the respective Gaussian processes, one must truncate infinite series, and potentially draw Monte Carlo samples from the respective probability measures.
These result in approximation error, which we analyze both analytically and empirically in this section.

\subsubsection{Asymptotics}
\label{sec:approx:asymptotics}

\paragraph{Truncation error on general manifolds.}
We begin with understanding truncation error, since except for very specific cases where the series can be summed analytically,\footnote{The only such case known to the authors is the circle $\bb{S}_1 = \SO(2)$, see Example 8 in \textcite{borovitskiy2020}, which adapts the argument in the supplementary material of \textcite{guinness2016}.} this approximation error is always present.
Truncation error is governed by the rate of convergence of the respective series, which depends in turn on the growth rate of the Laplace--Beltrami eigenvalues $\alpha_\lambda$, the representation degrees $d_\lambda$, as well as $r_\lambda$ for homogeneous spaces, and on the kernel's parameters.
Throughout this section we work on a compact homogeneous space $X$.
However, to get the basic asymptotics, we initially ignore the homogeneous space structure of $X$, and treat it as a general compact Riemannian manifold: then, $k_{\nu, \kappa, \sigma^2}$ can be obtained via the general manifold Fourier feature expansion.
Let
\[
\label{eqn:truncated_series}
k_{\nu, \kappa, \sigma^2}^{(J)}(x, x')
&=
\frac{\sigma^2}{C_{\nu, \kappa}}
\sum_{j=0}^{J} \Psi_{\nu, \kappa}(\lambda_j) f_j(x) f_j(x')
&
\Psi_{\nu,\kappa}(\lambda)
&=
\begin{cases}
\del{\frac{2 \nu}{\kappa^2} + \lambda}^{-\nu-n/2} & \nu < \infty
\\
e^{-\frac{\kappa^2}{2} \lambda} & \nu = \infty
\end{cases}
\]
where $(-\lambda_j, f_j)$ are Laplace--Beltrami eigenpairs, defined such that $\cbr{f_j}_{j=0}^{\infty}$ is an orthonormal basis of $L^2(X)$ and $n = \dim(X)$.
Note that $k_{\nu, \kappa, \sigma^2}^{(\infty)} =  k_{\nu, \kappa, \sigma^2}$ where the right-hand side is the Matérn or heat kernel, as discussed in~\Cref{sec:heat_matern:heat,sec:heat_matern:matern}.
For simplicity, we set $\sigma^2 = C_{\nu, \kappa}$ and omit the corresponding index.
Since the presented kernel expressions for Lie groups and homogeneous spaces can be viewed as simplifications of this general expression, which group together and combine certain eigenfunctions into characters and zonal spherical functions, respectively, one can understand truncation error by analyzing the respective tail of \Cref{eqn:truncated_series}.

Specifically, we focus on the \emph{$L^2$-truncation error} $\norm[0]{k_{\nu, \kappa} - k_{\nu, \kappa}^{(J)}}_{L^2(X\x X)}$.
To bound this, we apply Weyl's Law \cite{chavel1984}, which for a general compact Riemannian manifold of dimension $n$ says that there is a constant $C > 0$ such that 
\[ \label{eqn:weyls_law}
C^{-1} j^{2/n} \leq \alpha_j \leq C j^{2/n}
.
\]
We obtain the following.
\begin{restatable}{proposition}{ThmTruncationBounds} \label{thm:truncation_bounds}
For $J \in \N$ sufficiently large, and $\tl{C}$ depending only on $\nu$, $\kappa$, and $C$, we~have
\[
\norm[1]{k_{\nu, \kappa} - k_{\nu, \kappa}^{(J)}}_{L^2(X \x X)}
\leq
\tl{C}
\begin{cases}
J^{-2\nu/n - 1/2} & \nu < \infty, \\
J^{1/2-1/n} e^{- \frac{\kappa^2}{2 C} J^{2/n}} & \nu = \infty.
\end{cases}
\]
\end{restatable}

\begin{proof}
\Cref{appdx:proofs}.
\end{proof}

This gives a bound on the truncation error, which is seen to decrease in a manner that depends on the kernel's smoothness as quantified by parameter $\nu$, and on the manifold's geometry through the growth rate of $\lambda_j$---in this simple case, determined by the dimension $n$ and constant $C$ only.
It is therefore clear that manifold Fourier feature truncation error decreases polynomially for Matérn kernels and exponentially for the heat kernel.

\paragraph{Truncation error on Lie groups.}
On a Lie group $X = G$, one can think of the character-based kernel expression as a simplification of the general manifold expression, where certain terms corresponding to repeated eigenvalues are combined.
Recall the expression 
\[ 
\label{eqn:character_series}
k^{(L)}_{\nu, \kappa}(g_1, g_2)
=
\sum_{\lambda \in \tl\Lambda}
a^{(\lambda)}
\,
\Re\chi^{(\lambda)}(g_2^{-1} \bdot g_1)
\]
of \Cref{sec:computation}, where $\tl\Lambda$ is the truncated set of $L$ irreducible representations, constructed according to \Cref{sec:computation:char}.
Note that each summand in this expression corresponds to $d_\lambda^2$ eigenfunctions in the manifold Fourier feature expansion.
This means that this series converges faster than manifold Fourier features do, with the precise speedup determined by the behavior of $d_\lambda$ for the group of interest.

To get an estimate of $d_\lambda$, one can apply the results of \textcite{augarten2020}, obtaining that the number of irreducible representations of dimension at most $N$ is at most $(r!)^3 N^{r/R}$, where $r$ is the rank of $G$ and $R$ is the number of positive roots of $G$, which grows quadratically with $r$. 
This reveals the Lie group's \emph{rank} to be the primary factor influencing the degree of speedup: the larger it is, the bigger $d_{\lambda}$ will be.
In general, however, understanding how geometric factors affect the precise interplay between the growth of $a_{\lambda}$ and $d_{\lambda}$ is non-trivial, and we therefore examine it empirically in the sequel.

\paragraph{Truncation error on homogeneous spaces.}

On a homogeneous space $X = G/H$ the situation is similar: recall that the zonal spherical functions play the role of characters, and
\[
k_{\nu, \kappa}^{(L)}(g_1 \bdot H, g_2 \bdot H)
=
\sum_{\lambda \in \tl\Lambda}
a^{(\lambda)}
\,
\sum_{j = 1}^{r_\lambda}
\Re \pi^{(\lambda)}_{j j}(g_2^{-1} \bdot g_1)
\]
where $\tl\Lambda$ is as above and $\pi^{(\lambda)}_{j j}$ are zonal spherical functions.
Each inner sum of $r_{\lambda}$ terms corresponds to $r_{\lambda} d_{\lambda}$ eigenfunctions, which up to a constant coincide with spherical functions.
Here, the rate of speedup is determined by the behavior of $d_\lambda$ and $r_\lambda$ for the particular homogeneous space under study.

\paragraph{Other asymptotics.}
The generalized periodic summation of \Cref{sec:computation:zonal} and generalized random phase Fourier feature estimators of \Cref{sec:computation:comp_pointwise,sec:computation:comp_sampling} are Monte Carlo estimators.
Therefore, they converge at the usual $\sqrt{S}$ rate, in the Monte Carlo sense, where $S$ is the number of samples.
The constants are governed by the interplay between the geometry of the space under consideration and properties of the kernel.
To understand how these settings perform and get a more quantitative view of truncation, we now turn to empirical~evaluation.

\subsubsection{Empirical Evaluation}
\label{sec:approx:empirical}

\begin{figure}[t]
\includegraphics{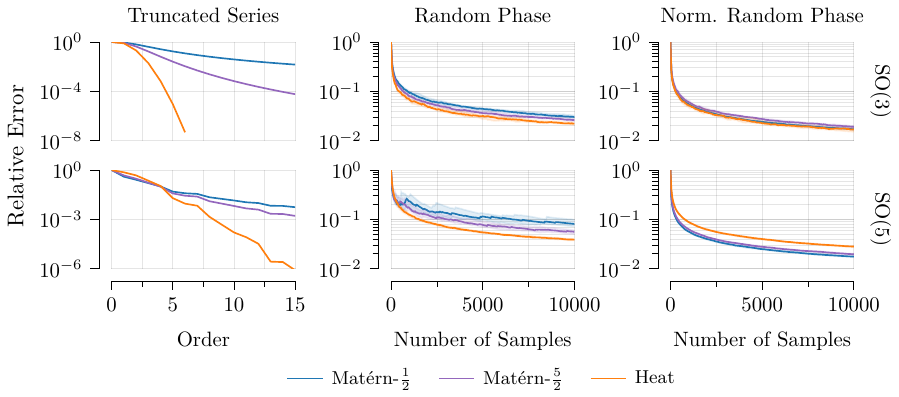}
\caption{Here we compute the approximation error for the truncated character-based series expression, and random phase approximation, as a function of the truncation order, and number of Monte Carlo samples, respectively.
For the random phase variants, we plot the median, as well as $25\%$ and $75\%$ quantiles, computed using $20$ different random seeds.
We examine relative error, which is equivalent to normalizing each curve so that it equals one at the origin.
We see that, for the character-based expansion, approximation error decreases rapidly with the truncation order, especially for the smoother heat kernel.
Note that the rate also depends on the length scale.
Random-phase approximations converge at a slower rate.
Normalized random phase approximations outperform their non-normalized~variants.
}
\label{fig:approx-error-rffs}
\end{figure}

We now evaluate the error in all of the proposed kernel approximations---based on truncation, periodic summation, and generalized random phase Fourier features---empirically.
For this, we selected a total of two Lie groups, namely the special orthogonal groups $\SO(3)$ and $\SO(5)$, and three homogeneous spaces, namely the Stiefel manifolds $\f{V}(1, 5)$, $\f{V}(2, 5)$, and $\f{V}(3, 5)$.
We consider three kernels of varying smoothness, namely a Matérn-$\frac{1}{2}$ kernel corresponding to a non-differentiable Gaussian process, a Matérn-$\frac{5}{2}$ kernel corresponding to a twice-differentiable Gaussian process, and heat kernel corresponding to an infinitely-differentiable Gaussian process.
We select length scales that ensure the processes are neither approximately spatially constant, nor resemble white noise.
All experiments---where applicable due to stochasticity---were repeated $20$ times to assess variability.

In our first experiment, we evaluate the truncation error of the character-based kernel expression, and the approximation error of the generalized random phase approximation, on $\SO(3)$ and $\SO(5)$.
For this, we first precompute the kernel $k_{\nu, \kappa}^{(L_{\max})}$ using $L_{\max} = 50$ total characters, to create a suitable set of ground-truth values for the kernel.
We assume $\sigma^2$ is chosen in such a way that $k_{\nu, \kappa}^{(L_{\max})}(x, x) = 1$.
We then vary the truncation order $L$, namely the number of terms used in the series in~\Cref{eqn:character_series}, and compare the resulting kernel $k_{\nu, \kappa}^{(L)}$ to the precomputed kernel.
We similarly vary the number of Monte Carlo samples $S$ used for the random phase approximation given by~\Cref{eqn:random_phase_lie} with number of characters fixed to $L_{\max}$.
We denote the covariance function of the resulting approximation by $k_{\nu, \kappa}^{(L,S)}$.
We also consider an alternative approximation $\tilde{k}_{\nu, \kappa}^{(L,S)}(x, x')$, which we call the \emph{normalized random phase} approximation, which forces $\tilde{k}_{\nu, \kappa}^{(L,S)}(x, x) = 1$. 
This approximation is defined by
\[
\tilde{k}_{\nu, \kappa}^{(L,S)}(x, x') = \frac{k_{\nu, \kappa}^{(L,S)}(x, x')}{\sqrt{k_{\nu, \kappa}^{(L,S)}(x, x) k_{\nu, \kappa}^{(L,S)}(x', x')}} 
.
\]
In all cases, we examine the $L^2$-error $\norm{k_{\nu, \kappa}^{L_{\max}} - \.}_{L^2(X \x X)}$ where $(\.)$ is either $k_{\nu, \kappa}^{(L)}$, $k_{\nu, \kappa}^{(L,S)}$ or $\tilde{k}_{\nu, \kappa}^{(L,S)}$.
In the former case, this error can be computed analytically in terms of $\Psi_{\nu, \kappa}(a^{(\lambda)})$.
In the latter cases, we draw $50$ uniformly random points $x_1, \ldots, x_{50} \in X$ and approximate
\[ \label{eqn:approx_norm}
\norm{k_{\nu, \kappa}^{(L_{\max})} - k_{\nu, \kappa}^{(L,S)}}_{L^2(X \x X)}^2
\approx
\frac{1}{50^2}
\sum_{i=1}^{50} \sum_{j=1}^{50}
\abs{k_{\nu, \kappa}^{(L_{\max})}(x_i, x_j) - k_{\nu, \kappa}^{(L,S)}(x_i, x_j)}^2
\]
and similarly for $\tilde{k}_{\nu, \kappa}^{(L,S)}$ instead of $k_{\nu, \kappa}^{(L,S)}$.

Results are given in \Cref{fig:approx-error-rffs}.
We see that, mirroring the theoretical picture given in \Cref{sec:approx:asymptotics}, the truncation error decay rate varies according to the smoothness of the kernel.
For the heat kernel, we observe superexponential decay in truncation error.
For the two Matérn kernels, we observe decay that is slightly slower than exponential on the scales considered.
This indicates that the extra $d_\lambda$-factor of the character-based formula produces a significant increase in accuracy on these manifolds, reinforcing the theoretical picture presented previously.

Next, we examine the random phase approximation, again on $\SO(3)$ and $\SO(5)$, whose results are also given in \Cref{fig:approx-error-rffs}.
Here, we see that error decreases at the expected Monte Carlo rate.
Compared to the previous setting, we also see less variability in behavior among kernels with different levels of smoothness.
In all cases, the normalized approximation $\tilde{k}_{\nu, \kappa}^L$ performs better: when applicable, we thus recommend using it in practice.

\begin{figure}[t]
\includegraphics{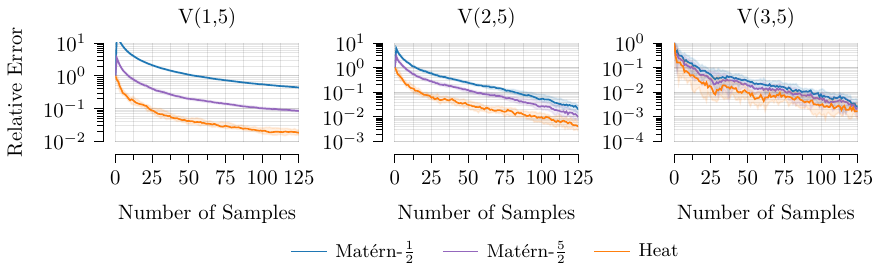}
\caption{
Here we compute the approximation error for the periodic-summation-based kernel computation technique.
We plot the median, as well as $25\%$ and $75\%$ quantiles, computed using $20$ different random seeds.
We examine relative error, which is equivalent to normalizing each curve so that it equals one at the origin.
We see that error decreases at a modest $\sqrt{S}$ rate, as one would expect given the Monte Carlo nature of the approximation.
}
\label{fig:approx-error-psum}
\end{figure}

For our next experiment, we evaluate the performance of periodic summation on the Stiefel manifolds $\f{V}(1, 5)$, $\f{V}(2, 5)$, and $\f{V}(3, 5)$.
Note that, of these, $\f{V}(1, 5) = \bb{S}_4$ is the four-dimensional sphere.
Similarly to the preceding experiment, we compute a ground-truth kernel: on $\f{V}(1, 5)$ we do so using the zonal-spherical-function-based available for the sphere\footnote{We use the first 10 eigenspaces on sphere. Note that from \textcite{gelbart1974}, these correspond to eigenspaces on $\SO(5)$ with signature $(k,0)$---we use this to improve the kernel approximation.}, whereas on $\f{V}(2, 5)$ and $\f{V}(3, 5)$ where such an expansion is not available we instead compute a ground-truth kernel using $150$ Monte Carlo samples.
We then vary the number of Monte Carlo samples used for periodic summation, and examine the $L^2$-error between the approximate kernel and the ground truth, again computed as in~\Cref{eqn:approx_norm}, but this time using $20$ random points instead of $50$.

Results are given in \Cref{fig:approx-error-psum}.
Here, we see that on $\f{V}(1, 5)$, approximation accuracy varies slightly according to smoothness, and decays slightly faster for the heat kernel compared to Matérn kernels.
On $\f{V}(2, 5)$ and $\f{V}(3, 5)$, no such differences are visible, and all techniques perform comparably---though, it is possible this similarity is caused by the comparatively-imprecise ground truth used in these cases.

In summary, we see that empirically, truncation error tends to decay very fast on the examples considered: a few dozen characters or zonal spherical functions is usually more than enough to accurately approximate the limiting kernel.
The Monte-Carlo-based approximations considered---generalized random phase Fourier features and generalized periodic summation---decay at the usual Monte Carlo rate.
For these, a reasonable number of samples may be on the order of hundreds or thousands to accurately represent the limiting kernel.
Finally, note that even if the error is large, all of the kernel approximations considered here are positive semi-definite, and therefore define perfectly valid and usable Gaussian processes in their own right, irrespective of the truncation level chosen.

\subsubsection{Computational Complexity}

We conclude by summarizing the computational costs of the three proposed kernel approximations: (i) truncated series expansions based on characters and zonal spherical functions, (ii) generalized periodic summation, and (iii) generalized random phase Fourier~features.

Given a set of truncated signatures $\tl\Lambda$ of size $L$, our proposed truncated series expansions are linear combinations of either characters or zonal spherical functions, depending on the setting.
The computational cost of evaluating the truncated series, given that characters or zonal spherical functions are known and cost $\c{O}(1)$ to evaluate, is therefore linear: this gives a complexity of $\c{O}(L)$.
In cases where these respective functions are known, they are usually either polynomials or ratios of such, with powers increasing as $\c{O}(L)$ grows.
Thus, the $\c{O}(1)$ assumption is arguably only appropriate in cases where these expansions are very efficient, and can be evaluated in a numerically stable manner using only lower-order terms.\footnote{In theory, there exists a character formula which can be calculated in constant time: unfortunately, this formula involves ratios where the denominator can evaluate to zero, rendering the calculation unstable. We therefore use an alternative formula, which leads to an explicit polynomial and is numerically stable, but whose size unfortunately grows with $L$.}
Fortunately, in practice this is often the case: in all of the examples we considered, a few dozen terms sufficed to obtain an accurate approximation.

Generalized periodic summation generally requires one to use Monte Carlo to evaluate the respective integral.
The computational costs of doing this directly---that is, without using a double sum to improve accuracy---are $\c{O}(LS)$, where $S$ is the number of samples.
Recall that, to guarantee positive semi-definiteness, one can modify the approximation by introducing a double sum: this results in $\c{O}(LS^2)$ kernel evaluation cost.
This is the most costly algorithm in this work, and shows that having an explicit form for the zonal spherical functions can be significantly beneficial in terms of efficiency.

Generalized random phase Fourier features provide a feature map with $LS$ features in total, which can be sampled in $\c{O}(S)$ time, where $S$ is the number of samples.
Evaluating this feature map costs $\c{O}(LS)$ for Lie groups, and $\c{O}(LRS)$ for homogeneous spaces, where $R = \max_{\lambda\in\tl\Lambda} r_\lambda$.
The complexity of approximately sampling a Gaussian process prior using this technique is the same.

\section{Conclusion}

In this work, we derived a set of techniques that make it possible to deploy Gaussian process models on a number of compact Riemannian manifolds which admit analytic descriptions and are important in applications.
This was done by combining and synthesizing results from a variety of mathematical fields, including probability, statistics, differential geometry, and representation theory.
Our techniques yield approximations which are essentially-closed-form algebraic expressions, and are guaranteed positive semi-definite, making them simple and reliable to use.
We hope our ideas enable practitioners to deploy these models in new prospective applications.

\section*{Acknowledgments}
The authors are grateful to Prof. Mikhail Lifshits for providing useful feedback.
IA and AS were supported by RSF grant N\textsuperscript{\underline{o}}21-11-00047.
AT was supported by Cornell University, jointly via the Center for Data Science for Enterprise and Society, the College of Engineering, and the Ann S. Bowers College of Computing and Information Science.
VB was supported by an ETH Zürich Postdoctoral Fellowship.
We further acknowledge support from Huawei Research and Development.

\appendix

\section{Traversing Representations via Highest Weights and Signatures} \label{appdx:trav_repr}

Here, primarily following \textcite{brockerdieck1985,fegan1991}, we review the relevant notions of representation theory and describe a general approach that allows one to traverse the set of irreducible unitary representations.

Recall that a substantial part of the mathematical structure of a Lie group is encoded within its tangent spaces. 
Due to the group structure, these are linked rigidly enough that it suffices to study the tangent space at one point, such as the identity element $e$ of the group. This tangent space, equipped with some additional algebraic structure, is called the \emph{Lie algebra} $\fr{g}$ of the group $G$.\footnote{A Lie algebra $\fr{g}$ can also be viewed as the set of left-invariant vector fields on $G$: using this perspective, we define the Lie bracket $[\.,\.] : \fr{g} \x \fr{g} \-> \fr{g}$ as the commutator of vector fields $[x,y] = xy - yx$. More generally, we define an abstract Lie algebra to be a finite-dimensional vector space equipped with an antisymmetric bilinear binary operation satisfying the Jacobi identity $[x,[y,z]]+[y,[z,x]]+[z,[x,y]]=0$}

Following the discussion from the main text, every irreducible representation of a connected compact semi-simple Lie group $G$ produces a representation of its maximal torus $T_m \subseteq G$ by restriction.
Using the structure of the torus, this representation splits into a sum of one-dimensional irreducible representations, parametrized by tuples of integers, which are called the \emph{weights} of the representation.
It turns out these representations possess even more structure: they are in fact uniquely determined by just one of the tuples, called the \emph{highest weight}.

We proceed to define the above-mentioned weights, along with the related notion of the \emph{roots} of a Lie group in an abstract and general fashion, which produces a configuration of vectors in $\fr{t}^*$, the space of linear functionals on the Lie algebra $\fr{t}$ of a fixed maximal torus.
This combinatorial structure enables us to enumerate all irreducible representations.

\subsection{Weights and Roots of Lie groups}

Let $G$ be a Lie group, and let $\fr{g}$ be its Lie algebra.
Usually $G$ is a linear group, namely it is realized as a closed subgroup of $\GL(V)$ for some finite-dimensional complex vector space $V$, in which case $\fr{g}$ is naturally realized as a subspace of the space of endomorphisms $\f{End}(V)$, defined as the space of all linear operators from $V$ to $V$.
In what follows, we use $\cl{x}$ and $\cl{A}$, respectively, to denote the conjugate of a complex number, and the closure of a subset of a topological space, with the intended meaning clear by context.
We now introduce the notion of a maximal torus.

\begin{result} \label{thm:computation:maxtorus}
Every compact Lie group $G$ contains a torus $T$ which is maximal with respect to inclusion. 
Any other torus is contained in a maximal one.
Maximal tori are not unique, but are conjugate to one another.
\end{result}

\begin{proof}
\textcite[Chapter IV]{brockerdieck1985}.
\end{proof}

We can use this to study representations of $G$.
Let $\fr{t}$ be the Lie algebra of $T$.
Given a representation $\pi: G\->\GL(V)$ of $G$ on a vector space $V$, consider its restriction $\tl\pi = \pi|_T$ to $T$. 
Since $T$ is a torus, this restriction decomposes into a direct sum of one-dimensional (complex-valued) irreducible unitary representations of $T$, given as
\[
\tl\pi(x) = \f{diag}(1, \ldots, 1, \exp(\theta_1(x)), \ldots, \exp(\theta_k(x)))
\]
where $\theta_i$ are homomorphisms $T\->\bb{T} \subset \C$.
We define the \emph{weights} of a representation in terms of these homomorphisms as follows.

\begin{definition}
Define the \emph{weights} of a representation $\pi : G \-> \GL(V)$ to be the real differentials $\d_e \theta_i:\fr{t}\->\R$, where $\theta_i$ are defined above.
By linearity, we have $\d_e \theta_i \in \fr{t}^*$ where $\fr{t}^*$ is the space of linear functionals over the space $\fr{t}$.
\end{definition}

For each $x\in G$, define the \emph{adjoint action} $I_x: G\-> G$ by $I_x(g)=x\bdot g \bdot x^{-1}$, and the \emph{adjoint representation} $\Ad: G\->\GL(\fr{g})$ by $\Ad(x) = \d_e I_x$.
If $G$ is a linear group, then the adjoint action is realized by conjugation of matrices in $\fr{g}$.
This allows us to introduce the \emph{roots} of $G$.

\begin{definition}
Define the \emph{roots} of a Lie group $G$ to be the non-zero weights of the adjoint representation $\Ad : G \-> \GL(\fr{g})$.
\end{definition}

The adjoint representation gives rise to a particular bilinear form on $\fr{g}$ that we will need further, called the \emph{Killing form}. 
We define it now.
Let $\ad : \fr{g} \-> \f{End}(\fr{g})$ be defined by $\ad(X)(Y) = [X,Y]$.

\begin{definition}
Define the symmetric bilinear \emph{Killing form} on the Lie algebra $\fr{g}$ by $B(X,Y) = \tr(\ad(X)\after\ad(Y))$ for $X, Y \in \fr{g}$, where $\after$ denotes composition.
We say $G$ is \emph{semi-simple} if $B$ is non-degenerate in the sense that $B(X, Y) = 0$ for all $Y$ implies $X = 0$.
\end{definition}

It is not hard to see that the Killing form is invariant under the adjoint action, that is,
\[
B(\Ad(g)X,\Ad(g)Y) = B(X,Y).
\]
Different choices of maximal tori $T^{(1)}$ and $T^{(2)}$ produce different sets of weights in the respective tangents spaces $\fr{t}^{(1)}$ and $\fr{t}^{(2)}$, but since all maximal tori are conjugate, the corresponding differential $\fr{t}^{(1)}\->\fr{t}^{(2)}$ is an isometry with respect to the Killing form, sending one set of the weights into another.

One can show that a connected semi-simple Lie group is compact if and only if $B$ is negative definite.
Reinterpreted as a tensor field, $-B$ therefore canonically defines a bi-invariant Riemannian structure.
Hereinafter by default we assume that the groups we are working with are semi-simple.

\subsection{The Highest Weight of a Representation} \label{sec:computation:comp_traverse:highest}

To enumerate representations we will use certain weights called the \emph{highest weights}.
For a given representation its highest weight is the maximal weight with respect to a certain partial order.
We proceed to define and formalize these concepts.

\begin{definition} \label{dfn:computation:la}
For each root $\alpha$, let $L_\alpha\subset\fr{t}$ be the hyperplane 
\[
L_\alpha = \{ X\in\fr{t} : \alpha(X)=0 \}
\]
which, by virtue of $\alpha$ being linear, passes through the origin.
Then $\fr{t}\setminus\cup_\alpha L_\alpha$ is by construction a union of disjoint convex polyhedral cones, each of which we call a \emph{Weyl chamber}.
\end{definition}

Define $\sigma:\fr{t}\->\fr{t}^*$ according to $-B(X,Y) = \sigma(X)(Y)$, where $B$ is the Killing form.
By non-degeneracy of $B$ this establishes an isomorphism between $\fr{t}$ and $\fr{t}^*$.
The subsets of $\fr{t}^*$ that are images of Weyl chambers under $\sigma$ will also be called Weyl chambers.

To define the ordering, we fix one arbitrary Weyl chamber $\fr{t}_+ \subseteq \fr{t}$ and call it \emph{the positive chamber}.
We also denote $\fr{t}^{*}_+ = \sigma(\fr{t}_+) \subseteq \fr{t}^*$.
This defines the order on $\fr{t}^*$ as follows.

\begin{definition}
Let $\theta,\theta'\in\fr{t}^*$.
We say that $\theta\geq\theta'$ if $\theta - \theta' \in \cl{\fr{t}^*_+}$, which defines a partial order.
We say that $\theta\in\fr{t}^*$ is \emph{positive} if $\theta(X)>0$ for any $X\in \fr{t}_+$.
\end{definition}

For purposes of calculating Laplace--Beltrami eigenvalues on compact Lie groups, we now introduce a particular weight $\rho$.

\begin{definition} \label{def:roots-half-sum}
Define the \emph{half-sum of the roots} to be $\rho = \frac{1}{2} \sum r$, where the sum is taken over all positive roots.
\end{definition}

Before proceeding to the main result on highest weights, we describe a certain algebraic structure generated by weights.
Let $\exp : \fr{g} \-> G$ be the Lie-theoretic exponential map.

\begin{definition}
Define the \emph{weight lattice}\footnote{We define a \emph{lattice} to be a discrete additive subgroup of $\R^m$.} to be
\[ 
P = \cbr{ c_1\alpha_1+\ldots+ c_k\alpha_k :\ k \in \N, \ c_j\in\Z,\ \t{and} \alpha_j\ \t{are weights}}
,
\]
which is spanned by the weights of all representations.
\end{definition}

With these notions in place, we can state the key structural result which we will use to enumerate the irreducible representations of $G$.

\begin{result}
For a given choice of a positive chamber $\fr{t}_+$, every representation $\pi$ of $G$ has a unique weight $w$ which is largest with respect to the ordering on $\fr{t}^*$ induced by $\fr{t}_+$.
This weight $w$ is always an element of $\cl{\fr{t}^*_+}$, and we call it the \emph{highest weight}.
Moreover, every element of $P\^ \cl{\fr{t}^*_+}$ is the highest weight of a unique irreducible representation.
\end{result}

\begin{proof}
See \textcite[Definition 2.7, Chapter VI]{brockerdieck1985} or \textcite[Theorems 8.19 and 8.21]{fegan1991}.
\end{proof}

This gives us a way to explicitly enumerate the irreducible representations of a given Lie group $G$: (1) choose a maximal torus $T \subseteq G$, (2) choose a positive chamber $\fr{t}_+ \subseteq \fr{t}$, where $\fr{t}$ is the Lie algebra of $T$, and (3) calculate the intersection of the weight lattice $P$ with the closure of the dual positive chamber $\cl{\fr{t}^*_+}$.
Every weight in this intersection maps uniquely onto an irreducible representation.

This establishes what actually needs to be done to calculate the index set $\Lambda$ of all irreducible representations. 
However, there is a more convenient way to carry this process out numerically.
If we identify our maximal torus $T$  with $\bb{T}^m \subseteq \C^m$, it turns out we can express the action of the highest weight of a representation simply by raising an element of $\bb{T}^m$ to the power of some integer multi-index.
This idea leads to the notion of a \emph{signature}, formulated as follows.

\begin{result}
\label{res:signature}
Let $w$ be the highest weight of a representation $\pi$.
Suppose the chosen maximal torus $T \subseteq G$ is isomorphic to $\bb{T}^m \subseteq \C^m$ via the isomorphism of form
\[
t \in T \-> (\varepsilon_1(t), \ldots, \varepsilon_m(t)) \in \bb{T}^m.
\]
Then for every $t \in T$, and every $X \in \fr{t}$ satisfying $t = \exp(X)$, $w$ can be expressed as
\[
w(X) = \varepsilon_1(t)^{p_1} \cdot \ldots \cdot \varepsilon_m(t)^{p_m}
.
\]
We call $(p_1,\ldots,p_m) \in \Z^m$ the \emph{signature} of the representation, which is unique up to a choice of a maximal torus $T$ and its isomorphism with the standard torus $\bb{T}^m$.
\end{result}

\begin{proof}
See \textcite[Chapter VII, Section 2]{naimark1982representations}.
\end{proof}

Thus, we can use signatures to identify highest weights with tuples of integers. 
Not all tuples will correspond to highest weights: various restrictions on the integers forming a signature vary according to $G$.
These are best handled on a case by case basis: see Appendix~\ref{appdx:weights_examples} for examples.

\section{Calculating Signatures and Highest Weights for \texorpdfstring{$\SL(n)$}{SL(n)}, \texorpdfstring{$\SU(n)$}{SU(n)}, and \texorpdfstring{$\SO(n)$}{SO(n)}} \label{appdx:weights_examples}

To illustrate the preceding ideas via example, we now calculate the above-introduced notions for three widely-occurring Lie groups.
We defer the remaining task of determining which irreducible representations corresponds to which highest weights to subsequent sections.

\begin{example}
If $G=\SL(n)$, then
\[
\fr{g} = \{ X\in M(n,\R) : \tr(X) = 0 \}.
\]
The most natural choice of a maximal torus $T$ is the subgroup of all diagonal matrices, then $\fr{t} \subset \fr{g}$ is naturally realized as the subspace of diagonal matrices. The Killing form on $\fr{g}$ is realized as $B(X,Y) = 2n\tr(XY)$, and its restriction to $\fr{t}$ is then
\[ B(\f{diag}(x_1,\ldots,x_n),\f{diag}(y_1,\ldots,y_n)) = 2n(x_1y_1+\ldots+x_ny_n). \]
The space $\fr{t}^*$, equipped with the symmetric bilinear form, is naturally isomorphic to the $n-1$-dimensional Euclidean space.
Denote $e_i=(0, \ldots, 1, \ldots, 0)$ where $1$ is in the $i$th component.
A standard realization of roots and weights is as follows. Consider the subspace $E$ of $\R^d$ orthogonal to $e_1+\ldots+e_n$. The roots are precisely the vectors $\alpha\in E$ having integer coordinates and $|\alpha|^2=2$. In other words, they are of the form $\alpha = e_i-e_j$ for $i\neq j$. The root lattice can be described as $Q = \Z^n\cap E$.

The half-sum of the positive roots then equals
\[
\rho = \del{\frac{n-1}{2}, \frac{n-3}{2}, \ldots, \frac{3-n}{2}, \frac{1-n}{2}}
.
\]

The weight lattice $P$ can be described as the dual of $Q$. Namely, one can find such $\varpi_i\in E$, $i=1,\ldots,n-1$, that $(\varpi_i, e_j-e_{j+1}) = \delta_{ij}$, then $Q = \sum_i \Z \varpi_i $.

A natural choice for a positive chamber $\cl{\fr{t}^*_+}$ is then the positive cone spanned by $\varpi_i$, namely the set of their linear combinations with non-negative coefficients. It can also be very conveniently described as
\[
\cl{\fr{t}^*_+} = \{ (x_1,\ldots,x_n) \in E\subset \R^n : x_1\geqslant x_2\geqslant \ldots \geqslant x_n \}.
\]

The action of the Weyl group is by all possible permutations of coordinates, thus $W \cong S_n$, the symmetric group on $n$ points.
\end{example}

\begin{example}
If $G=\SU(n)$, then its tangent space is
\[ \fr{g} = \cbr{ X\in M(n,\mathbb{C}) : X^*=-X,\ \tr(X)=0 }. \]
The most natural choice of a maximal torus is the subgroup of diagonal matrices, that is,
\[ T = \cbr{ \operatorname{diag}\del{ e^{i\theta_1}, \ldots, e^{i\theta_n} } : \theta_1+\ldots+\theta_n=0}. \]
Its tangent space is then
\[ \fr{t} = \cbr{ \operatorname{diag}\del{ i\theta_1, \ldots, i\theta_n } : \theta_1+\ldots+\theta_n=0 }. \]
The Killing form on $\fr{t}$ can be calculated as
\[ B(i\operatorname{diag}(\theta_1,\ldots,\theta_n), i\operatorname{diag}(\eta_1,\ldots,\eta_n)) = -2n(\theta_1\eta_1+\ldots+\theta_n\eta_n). \]
Under the identification $\operatorname{diag}\del{ i\theta_1, \ldots, i\theta_n } \mapsto (\theta_1,\ldots,\theta_n)$ of $\fr{t}$ with $\R^n$ the description of roots and weights coincides with that of $\SL(n)$.
\end{example}

\begin{example}
Let $G=\SO(n)$, its tangent is $\fr{g} = \cbr{ X\in \M(n,\mathbb{R}) : A^\top=-A }$.

The description of roots and weights for $G=\SO(n)$ depends on the oddity of $n$. For $n=2k$ or $2k+1$ the maximal tori $T$ can be chosen as
\[ 
T &= \cbr{
\begin{pmatrix}
\cos\theta_1 & \sin\theta_1 \\
-\sin\theta_1 & \cos\theta_1 \\
&& \ddots \\
&&& \cos\theta_k & \sin\theta_k \\
&&& -\sin\theta_k & \cos\theta_k
\end{pmatrix}
}
\\
T &= \cbr{
\begin{pmatrix}
\cos\theta_1 & \sin\theta_1 \\
-\sin\theta_1 & \cos\theta_1 \\
&& \ddots \\
&&& \cos\theta_k & \sin\theta_k \\
&&& -\sin\theta_k & \cos\theta_k \\
&&&&& 1
\end{pmatrix}
}
.
\]
Their tangent spaces are then
\[
\fr{t} &= \cbr{
\begin{pmatrix}
0 & \theta_1 \\
-\theta_1 & 0 \\
&& \ddots \\
&&& 0 & \theta_k \\
&&& -\theta_k & 0
\end{pmatrix}
}
&
\fr{t} &= \cbr{
\begin{pmatrix}
0 & \theta_1 \\
-\theta_1 & 0 \\
&& \ddots \\
&&& 0 & \theta_k \\
&&& -\theta_k & 0 \\
&&&&& 0
\end{pmatrix}
}
\]
and the formulas for the Killing form are
\[
&\begin{aligned}
B\Big( \operatorname{bdiag}\del{ \begin{psmallmatrix} 0 & \theta_1 \\ -\theta_1 & 0 \end{psmallmatrix}, \ldots, \begin{psmallmatrix} 0 & \theta_k \\ -\theta_k & 0 \end{psmallmatrix}}, \operatorname{bdiag}&\del{\begin{psmallmatrix} 0 & \eta_1 \\ -\eta_1 & 0 \end{psmallmatrix}, \ldots, \begin{psmallmatrix} 0 & \eta_k \\ -\eta_k & 0 \end{psmallmatrix}} \Big)
\\
&= (4-4k)(\theta_1\eta_1+\ldots+\theta_k\eta_k)
\end{aligned}
\\
&\begin{aligned}
B\Big( \operatorname{bdiag}\del{ \begin{psmallmatrix} 0 & \theta_1 \\ -\theta_1 & 0 \end{psmallmatrix}, \ldots, \begin{psmallmatrix} 0 & \theta_k \\ -\theta_k & 0 \end{psmallmatrix}}, \operatorname{bdiag}&\del{ \begin{psmallmatrix} 0 & \eta_1 \\ -\eta_1 & 0 \end{psmallmatrix}, \ldots, \begin{psmallmatrix} 0 & \eta_k \\ -\eta_k & 0 \end{psmallmatrix} }, 0 \Big)
\\
&= (2-4k)(\theta_1\eta_1+\ldots+\theta_k\eta_k)
.
\end{aligned}
\]
In both cases, $\fr{t}^*$ is naturally identified with $\R^k$.
The roots are, then,
\[
e_i\pm e_j,\ i\neq j \quad \text{for} \quad n=2k, \\
\pm e_i\ \ \text{and}\ \ e_i\pm e_j,\ i\neq j \quad \text{for} \quad n=2k+1.
\]
The half-sum of the positive roots equals
\[
\rho = \del{ n-1, n-2, \ldots, 2, 1, 0 } \quad \text{for} \quad n=2k, \\
\rho = \del{ \frac{n}{2}-1, \frac{n}{2}-2, \ldots, \frac{3}{2}, \frac{1}{2} } \quad \text{for} \quad n=2k+1.
\]
The weight lattice $Q$ in both cases can be described as $\Z^k$. The action of the Weyl group is by permutations of coordinates and changing of signs.\footnote{In group-theoretic terms, $W \cong S_k \ltimes (\Z/2\Z)^k$, the semidirect product of $S_k$ acting on the $k$-th power of the cyclic group of order $2$ by permutations of components.}

The standard choice of the positive chamber is
\[
\cl{\fr{t}^*_+} = \{ (x_1,\ldots,x_k)\in Q : x_i\geqslant x_{i+1},\ x_{k-1}\geqslant |x_k|  \} \quad \text{for} \quad n=2k, \\
\cl{\fr{t}^*_+} = \{ (x_1,\ldots,x_k)\in Q : x_i\geqslant x_{i+1}\geqslant 0 \} \quad \text{for} \quad n=2k+1.
\]
\end{example}

\begin{example}
For $G=\SO(n)$ and the choices of $T$ and $\cl{\fr{t}^*_+}$ as above, the possible signatures $(p_1,\ldots,p_k)$ are
\[ 
p_1\geqslant p_2 \geqslant \ldots \geqslant p_{k-1} \geqslant |p_k| \quad \text{for} \quad n=2k, \\
p_1\geqslant p_2 \geqslant \ldots \geqslant p_k \geqslant 0 \quad \text{for} \quad n=2k+1.
\]
\end{example}

\begin{example}
For $G=\SU(n)$ the possible signatures $(p_1,\ldots,p_n)$ are those with
\[ p_1\geqslant p_2\geqslant\ldots\geqslant p_n. \]
\end{example}

\section{The Weyl Character Formula}
\label{appdx:character_formula}
To describe the Weyl character formula, we will need the principal group of symmetries of a root system, its Weyl group.
Again, denote by $T$ a fixed maximal torus and by $\fr{t}$ its Lie algebra.

\begin{definition}
Define the \emph{Weyl group} $W$ as the subgroup of $\f{O}(\fr{t})$ generated by all reflections in $L_\alpha$ of \Cref{dfn:computation:la}.
$W$ acts transitively on the set of all Weyl chambers by permutations.
\end{definition}

To proceed, we introduce two ways of expressing certain alternating sums.

\begin{result}[Weyl Denominator Formula]
Let $X \in \fr{t}$, and define
\[
j(\exp(X)) = \sum_{\sigma\in W} \sgn(\sigma)\exp(\sigma(\rho)(X))
\]
where $W$ is the Weyl group, $\rho$ is half the sum of positive roots of \Cref{def:roots-half-sum}, and $\sgn(\sigma)$ is the determinant of $\sigma$ as a linear operator.\footnote{Note that $\sgn(\sigma)$ is equal to $(-1)^m$ when $\sigma$ is a composition of $m$ reflections.}
Then 
\[
j(\exp(X)) = \prod_{\alpha>0} \del{\exp\del{\tfrac{i}{2}\alpha(X)} - \exp\del{\tfrac{-i}{2}\alpha(X)}}
\]
where the product is taken over all positive roots.
\end{result}

\begin{proof}
\textcite{fegan1991}, Theorem 9.2.
\end{proof}

We now introduce the Weyl character formula, which expresses characters as the ratio of two alternating sums depending on the chosen tangent vector.

\begin{result}[Weyl Character Formula]
For a highest weight $w$, define 
\[ 
j_w(\exp(X)) = \sum_{\sigma\in W} \sgn(\sigma)\exp(\sigma(w+\rho)(X))
.
\]
Then for $t\in T$, the character $\chi_w$ can be expressed as
\[ 
\chi_w(t) = \frac{j_w(t)}{j(t)}
.
\]
\end{result}

\begin{proof}
\textcite{fegan1991}, Theorem 9.9.
\end{proof}

Ostensibly, this gives us a way to calculate character values on the maximal torus $T$.
However, since characters are conjugation-invariant, the Weyl character formula can be used to calculate character values everywhere on $G$.
We state this formally as follows.

\begin{corollary}
Let $g \in G$, and let $t \in T$ be such that there is an $h \in G$ such that $g = h \bdot t \bdot h^{-1}$, where $t$ and $h$ always exist by \Cref{thm:computation:maxtorus}.
Then 
\[
\chi_w(g) = \chi_w(h^{-1} \bdot t \bdot h) = \chi_w(t)
.  
\]
\end{corollary}

\begin{proof}
This follows because the character is invariant under conjugation:
\[
\chi_w(h^{-1} \bdot t \bdot h)
=
\tr \pi_w(h^{-1}) \pi_w(t) \pi_w(h)
=
\tr \pi_w(h^{-1}) \pi_w(h) \pi_w(t)
=
\chi_w(t).
\]
\end{proof}

This establishes how to calculate the character $\chi_w$ for a given highest weight $w$.
We now study how to express this using signatures.
If $t=(\varepsilon_1,\ldots,\varepsilon_n)$ under the correspondence between the maximal torus and the complex exponentials, we can re-express the above formula as the ratio of two polynomials in $\varepsilon_1,\ldots,\varepsilon_n$, where the denominator divides the numerator, and the ratio itself is also a polynomial.

In many cases, this expression can be simplified greatly. 
As an example, for $G=\SU(2)$, the Weyl character formula simply says that for a unique representation of dimension $d$---with signature $(d)$---the character value equals
\[
\chi_d(g) = \frac{\lambda^d-\lambda^{-d}}{\lambda-\lambda^{-1}}
\]
where $\lambda,\lambda^{-1}$ are the eigenvalues of $g$. 
Since both eigenvalues can be expressed as 
\[
\frac{1}{2}\del{ \tr(g)\pm\sqrt{\tr(g)^2-4} }
\]
one can show that
\[
\chi_d(g) = U_{d-1}\del{\tfrac{1}{2}\tr(g)}
\]
where $U_k$ is the $k$th Chebyshev polynomial of the second kind.

In general, this expression can either be calculated symbolically from the Weyl character formula, or obtained from the Kostant multiplicity formula, which expresses the coefficients as certain combinatorial quantities related to the structure of the root system of $G$.
Note in particular that from this expression, one can deduce the Weyl dimension formula
\[ 
\dim V_w = \chi_w(e) = \frac{\prod_{\alpha>0} (w+\rho, \alpha)}{\prod_{\alpha>0} (\rho,\alpha)}
\]
which is used for calculating normalizing constants $d_\lambda$.
In most cases, the alternating sums appearing in the above cases can be interpreted with the help of Weyl denominator formula as certain matrix determinants.

\section{Calculating the Characters for \texorpdfstring{$\SU(n)$}{SU(n)}, and \texorpdfstring{$\SO(n)$}{SO(n)}} \label{appdx:particular_character_formulas}

In more practical terms, for $\SO(n)$ and $\SU(n)$ Weyl character formula provides an explicit expression for $\chi(g)$ as a ratio of two Vandermonde-type determinants depending only on the eigenvalues of $g$, possibly together with an ordering.

\begin{example}
Let $G=\SO(n)$ and let $(p_1,\ldots,p_n)$ be the signature of a representation with the highest weight $\varpi$. Then the corresponding character $\chi_\varpi$ can be calculated as follows.
If $n=2k+1$, denote $q_i = p_i+k-i+\frac{1}{2}$ and set $(\gamma_1,\ldots,\gamma_k,1,\gamma_k^{-1},\ldots,\gamma_1^{-1})$ to be the eigenvalues of an element $g\in G$. Then
\[ \chi_\varpi(g) = \frac{\xi_1(q_1,\ldots,q_k)}{\xi_1(k-\frac{1}{2},\ldots,\frac{3}{2},\frac{1}{2})}, \quad \text{where} \quad \xi_1(d_1,\ldots,d_k) = \det\del{ \gamma_i^{d_j} - \gamma_i^{-d_j} }_{i,j=1}^k. \]
Weyl's dimension formula specializes to
\[ d_\varpi = \frac{2^k}{(2k-1)!\ldots 3!\,1!} \cdot q_1\ldots q_k \cdot \prod_{\mathclap{1\leqslant i<j\leqslant k}} (q_i^2-q_j^2). \]

If $n=2k$, denote $q_i=p_i+k-i$ for $i=1,\ldots,k-1$ and $q_k=|p_k|$. Let the quantity $(\gamma_1,\ldots,\gamma_k,\gamma_k^{-1},\ldots,\gamma_1^{-1})$ be the eigenvalues of $g\in G$ ordered such that
\[
\begin{pmatrix}
\cos\theta_1 & \sin\theta_1 \\
-\sin\theta_1 & \cos\theta_1 \\
&& \ddots \\
&&& \cos\theta_k & \sin\theta_k \\
&&& -\sin\theta_k & \cos\theta_k \\
&&&&& 1
\end{pmatrix}, \quad \text{where} \quad \gamma_j = \cos\theta_j + i \sin\theta_j
\]
is conjugated in $SO(n)$ to $g$ (contrary to the case $n=2k+1$, here each possible collection of eigenvalues determine two conjugacy classes). Introduce
\[
\xi_0(d_1,\ldots,d_k) &= \det\del{ \gamma_i^{d_j} + \gamma_i^{-d_j} }_{i,j=1}^k
&
\xi_1(d_1,\ldots,d_k) &= \det\del{ \gamma_i^{d_j} - \gamma_i^{-d_j} }_{i,j=1}^k. 
\]
Then
\[
\chi_\varpi(g) = \frac{\xi_0(q_1,\ldots,q_k)}{\xi_0(m-1,m-2,\ldots,1,0)}, \quad \text{if} \quad p_k=q_k=0, \\
\chi_\varpi(g) = \frac{\xi_0(q_1,\ldots,q_k)+(\sgn p_m)\xi_1(q_1,\ldots,q_k)}{2\xi_0(m-1,m-2,\ldots,1,0)}, \quad \text{if} \quad p_m\neq 0.
\]
The dimension formula turns into
\[ d_\varpi = \frac{2^{k-1}}{(2k-2)!(2k-4)!\ldots 4!\,2!} \cdot \prod_{\mathclap{\quad 1\leqslant i<j\leqslant k}} (q_i^2-q_j^2). \]
\end{example}

\begin{example}
For $G=\SU(n)$ the character of a representation with the signature $(p_1,\ldots,p_n)$ can be calculated as follows. Denote $q_i = p_i+n-i$. If $(\gamma_1,\ldots,\gamma_n)$ are the eigenvalues of $g\in G$, then
\[ \chi(g) = \frac{\xi_1(q_1,\ldots,q_k)}{\xi_1(0,1,\ldots,n-1)}, \quad \text{where} \quad \xi_1(d_1,\ldots,d_n) = \det\del{ \gamma_i^{d_j} }_{i,j=1}^n. \]
\end{example}

\section{Proofs} \label{appdx:proofs}

\ThmStationaryGroup*

\begin{proof}
Theorem 2 of \textcite{yaglom1961} shows that for any stationary Gaussian process $k(g_1, g_2) = \sum_{\lambda \in \Lambda} a^{(\lambda)} \, \chi^{(\lambda)}(g_2^{-1} \bdot g_1)$ where $a^{(\lambda)} \geq 0$ satisfy $\sum_{\lambda\in\Lambda} d_{\lambda} a^{(\lambda)} < \infty$.
Under our usual assumption that a Gaussian process is real-valued, we have $\Re k(g_1, g_2) = k(g_1, g_2)$. Using this and taking the real parts of both sides of the equation, we get \Cref{eqn:stationary_lie:kernel}.

To prove positive semi-definiteness, write, using (i) the definition of $\chi^{(\lambda)}$, (ii) the homomorphism property of $\pi^{(\lambda)}$, (iii) the fact that $\pi^{(\lambda)}(g)^{-1} = \pi^{(\lambda)}(g)^*$ due to unitarity, where the star denotes the operation of taking the conjugate transpose matrix, and (iv) the standard properties of traces, that
\[
\chi^{(\lambda)}(g_2^{-1} g_1)
=
\tr \pi^{(\lambda)}(g_2^{-1} g_1)
=
\tr \pi^{(\lambda)}(g_1) \pi^{(\lambda)}(g_2)^*
=
\sum_{j = 1}^{d_{\lambda}} \sum_{k = 1}^{d_{\lambda}} \pi^{(\lambda)}_{j k}(g_1) \overline{\pi^{(\lambda)}_{j k}(g_2)}.
\]
This immediately shows that a character $\chi^{(\lambda)}$, being an inner product of the feature maps $g \-> \cbr[0]{\pi^{(\lambda)}_{j k}(g)}_{1 \leq j, k \leq d_{\lambda}}$, is (complex-valued) positive semi-definite.
The real part of a complex-valued positive semi-definite function is positive semi-definite, hence $\Re \chi^{(\lambda)}$ is positive semi-definite.
\end{proof}

\ThmStationaryHomogeneous*

\begin{proof}
Theorem 5 of \textcite{yaglom1961} immediately implies that any stationary real-valued Gaussian process is of form $k(g_1 \bdot H, g_2 \bdot H) = \sum_{\lambda \in \Lambda} \sum_{j, k=1}^{r_\lambda} a^{(\lambda)}_{j k} \pi^{(\lambda)}_{j k}(g_2^{-1} \bdot g_1) $ where the coefficients $a^{(\lambda)}_{j k} \in \C$ form Hermitian positive semi-definite matrices $\m{A}^{(\lambda)}$ of size $r_\lambda \x r_\lambda$ with \emph{complex} entries satisfying $\sum_\lambda \tr \m{A}^{(\lambda)} < \infty$, and $\pi^{(\lambda)}_{j k}$ are the zonal spherical functions.
Furthermore, any covariance of this form, provided it is real-valued, it distribution-wise uniquely defines a stationary Gaussian process.
To obtain the claim, we need to understand when $k$ defined in this way is a real-valued function.

The map $g \-> \del[0]{\pi^{(\lambda)}(g)}^*$ defines an irreducible unitary representation, thus is equal to $\pi^{(\lambda')}$ for some $\lambda' \in \Lambda$.
No matter if $\lambda' = \lambda$ or not, it is easy to see that $r_{\lambda'} = r_\lambda$ and, what is more, we can define the zonal spherical functions for $\pi^{(\lambda')}$ using the same basis $e_1, \ldots, e_{d_{\lambda}} \in V_{\lambda'} = V_{\lambda}$, getting
\[
\pi^{(\lambda')}_{j k}(g)
=
\innerprod[1]{\pi^{(\lambda')}(g) e_j}{e_k}_{V_{\lambda}}
=
\innerprod[1]{\del[0]{\pi^{(\lambda)}(g)}^* e_j}{e_k}_{V_{\lambda}}
=
\innerprod[1]{e_j}{\pi^{(\lambda)}(g) e_k}_{V_{\lambda}}
=
\overline{\pi^{(\lambda)}_{k j}(g)}.
\]
Using this property, from $\Re k = k$, or equivalently from $k = \overline{k}$ we infer
\[
\sum_{\lambda \in \Lambda}
\sum_{j k = 1}^{r_{\lambda}}
a^{(\lambda)}_{j k}
\pi^{(\lambda)}_{j k} (g)
=
\sum_{\lambda \in \Lambda}
\sum_{j k = 1}^{r_{\lambda}}
\overline{a^{(\lambda)}_{j k}
\pi^{(\lambda)}_{j k} (g)}
=
\sum_{\lambda \in \Lambda}
\sum_{j k = 1}^{r_{\lambda}}
\overline{a^{(\lambda')}_{k j}}
\pi^{(\lambda)}_{j k} (g)
\]
Because all $\pi^{(\lambda)}_{j k}$ are orthonormal, it follows that $a^{(\lambda)}_{j k} = \overline{a^{(\lambda')}_{k j}}$.

Now, since $k(g_1 \bdot H, g_2 \bdot H) = k(g_2 \bdot H, g_1 \bdot H)$, write
\[
\sum_{\lambda \in \Lambda}
\sum_{j k = 1}^{r_{\lambda}}
a^{(\lambda)}_{j k}
\pi^{(\lambda)}_{j k} (g)
&=
\sum_{\lambda \in \Lambda}
\sum_{j k = 1}^{r_{\lambda}}
a^{(\lambda)}_{j k}
\pi^{(\lambda)}_{j k} (g^{-1})
=
\sum_{\lambda \in \Lambda}
\sum_{j k = 1}^{r_{\lambda}}
a^{(\lambda)}_{j k}
\overline{\pi^{(\lambda)}_{k j} (g)}
\\
&=
\sum_{\lambda \in \Lambda}
\sum_{j k = 1}^{r_{\lambda}}
a^{(\lambda')}_{k j}
\pi^{(\lambda)}_{j k} (g)
\]
This implies $a^{(\lambda)}_{j k} = a^{(\lambda')}_{k j}$ and therefore $a^{(\lambda)}_{j k} \in \R$.
Using these observations, write
\[
\Re
\sum_{\lambda \in \Lambda}
\sum_{j, k=1}^{r_\lambda}
a^{(\lambda)}_{j k} \pi^{(\lambda)}_{j k}(g_2^{-1} \bdot g_1)
=
\sum_{\lambda \in \Lambda}
\sum_{j, k=1}^{r_\lambda}
a^{(\lambda)}_{j k} \Re \pi^{(\lambda)}_{j k}(g_2^{-1} \bdot g_1)
\]
which gives the first part of the claim.

We now show that for each individual $\lambda \in \Lambda$ the corresponding sum over $j, k$ is positive semi-definite. \textcite[Theorem 5]{yaglom1961} proves, as a byproduct, that this sum is a covariance function of a Gaussian process given by
\[
f_{\lambda}(g H)
=
\sum_{j=1}^{d_{\lambda}}
\sum_{k=1}^{r_\lambda}
z_{j k}^{\lambda} \pi^{(\lambda)}_{j k}(g)
\]
where $z_{j k}^{\lambda}$ are zero-mean Gaussians with covariances given by $\E z_{j_1 k_1}^{\lambda} z_{j_2 k_2}^{\lambda} = \delta_{j_1 j_2} a^{(\lambda)}_{k_1 k_2}$.
The claim follows.
\end{proof}

\ThmPeriodicSummation*

\begin{proof}
Recall that under a suitable choice of basis in the representation space $V_{\lambda}$, spherical functions $\pi^{(\lambda)}_{j k}$ are matrix coefficients with $1 \leq j \leq d_{\lambda}$ and $1 \leq k \leq r_{\lambda}$.

Consider the integrals of form
\[
I^{(\lambda)}_{j k}(g)
=
\int_H
\pi^{(\lambda)}_{j k}(g \bdot h)
\d \mu_H(h).
\]
Since $\pi^{(\lambda)}_{j k}$ corresponding to $1 \leq j \leq d_{\lambda}$ and $1 \leq k \leq r_{\lambda}$ are spherical functions and these are constant on cosets, we have $\pi^{(\lambda)}_{j k}(g \bdot h) = \pi^{(\lambda)}_{j k}(g)$ and thus $I^{(\lambda)}_{j k}(g) = \pi^{(\lambda)}_{j k}(g)$.
Furthermore, since spherical functions span the subspace of $\Span \pi^{(\lambda)}_{j k}$ that consists of all functions constant on cosets, and since the remaining matrix coefficients belong to its orthogonal complement, we have $I^{(\lambda)}_{j k}(g) = 0$ for all remaining indices $j, k$, proving that
\[
\int_H \chi^{(\lambda)}(g \bdot h)
\d \mu_H(h)
=
\sum_{j = 1}^{d_\lambda}
\int_H
\pi^{(\lambda)}_{j j}(g \bdot h)
\d \mu_H(h)
=
\sum_{j = 1}^{r_\lambda}
\pi^{(\lambda)}_{j j}(g \bdot h).
\]
The same holds for real parts.

With this, we can write
\[
k(g_1 \bdot H, g_2 \bdot H)
&=
\sum_{\lambda \in \Lambda}
a^{(\lambda)}
\sum_{j = 1}^{r_\lambda}
\Re \pi^{(\lambda)}_{j j}(g_2^{-1} \bdot g_1)
\\
&=
\sum_{\lambda \in \Lambda}
a^{(\lambda)}
\int_H \Re \chi^{(\lambda)}(g_2^{-1} \bdot g_1 \bdot h)
\d \mu_H(h)
\\
&=
\int_H
\sum_{\lambda \in \Lambda}
a^{(\lambda)}
\Re \chi^{(\lambda)}(g_2^{-1} \bdot g_1 \bdot h)
\d \mu_H(h)
\\
&=
\int_H
k_G(g_1 \bdot h, g_2)
\d \mu_H(h).
\]

Define $f'(gH) = \int_H f_G(g \bdot h) \d \mu_H(h)$ where $f_G \~[GP](0, k_G)$.
Now, $\E f'(g H) = 0$, while writing
\[
\Cov\del{f'(g_1 H), f'(g_2 H)}
&=
\int_H \int_H
\Cov\del{f_G(g_1 \bdot h_1), f_G(g_2 \bdot h_2)}
\d \mu_H(h_1) \d \mu_H(h_2)
\\
&=
\int_H \int_H
k_G(g_1 \bdot h_1, g_2 \bdot h_2)
\d \mu_H(h_1) \d \mu_H(h_2)
\\
&=
\int_H \int_H
k_G(g_1 \bdot h_1 \bdot h_2^{-1}, g_2)
\d \mu_H(h_1) \d \mu_H(h_2)
\\
&=
\int_H
k_G(g_1 \bdot h, g_2)
\d \mu_H(h)
=
k(g_1 \bdot H, g_2 \bdot H)
\]
proves the claim.
\end{proof}

\ThmKPhaseCond*

\begin{proof}
Expanding the right-hand side of~\Cref{eqn:rkhs_prop_to_test} we get
\[
\sum_{\substack{k_1, k_2 = 1 \\ k_3, k_4 = 1}}^{r_\lambda}
a^{(\lambda)}_{k_1, k_2}
a^{(\lambda)}_{k_3, k_4}
\sum_{j_1, j_2 = 1}^{d_{\lambda}}
            \pi^{(\lambda)}_{j_1, k_1}(x)
  \overline{\pi^{(\lambda)}_{j_2, k_3}(x')}
\obr{
\int_{G/H}
  \overline{\pi^{(\lambda)}_{j_1, k_2}(u)}
            \pi^{(\lambda)}_{j_2, k_4}(u)
\d\mu_{G/H}(u)
}^{d_{\lambda}^{-1} \delta_{j_1 j_2} \delta_{k_2 k_4}}
\\ \label{eqn:rkhs_prop_to_test_rhs}
=
d_{\lambda}
\sum_{k_1, k_3 = 1}^{r_\lambda}
\ubr{\del{
  \sum_{k = 1}^{r_\lambda}
  \frac{a^{(\lambda)}_{k_1, k}}{d_{\lambda}}
  \frac{a^{(\lambda)}_{k_3, k}}{d_{\lambda}}
}}_{\t{denote by}\, b_{k_1, k_3}}
\sum_{j}^{d_{\lambda}}
          \pi^{(\lambda)}_{j, k_1}(x)
\overline{\pi^{(\lambda)}_{j, k_3}(x')}.
\]
Comparing the last expression to the right-hand side of~\Cref{eqn:homogeneous_K_defn} we see that the necessary and sufficient condition for them to match is that $d_{\lambda} b_{k_1, k_2} = a_{k_1, k_2}$ for all $1 \leq k_1, k_2 \leq r_\lambda$.
If let $\m{B}$ be the matrix consisting of $b_{k_1, k_2}$, then by~\Cref{eqn:rkhs_prop_to_test_rhs} we have $\m{B} = \m{A} \m{A}^\top$, along with the preceding condition, which in matrix form is $\m{B} = \m{A}$.
Hence, the necessary and sufficient condition for~\Cref{eqn:rkhs_prop_to_test} to hold is $\m{A} = \m{A} \m{A}^{\top}$ and since $\m{A}$ is symmetric by assumption, namely $\m{A} = \m{A}^{\top}$, the condition simplifies to $\m{A} \m{A}^{\top} = \m{A}^2 = \m{A}$.
\end{proof}

\ThmFundamentalSetExists*

\begin{proof}
Induction. For $k=1$, the claim is immediate since $\phi_1 \neq 0$.
Now, assume $\det \m{M}_k \neq 0$ for all $k$ in the range $k = 1, \ldots, n - 1$.
Split the matrix $\m{M}_n$ into blocks as follows:
\[
\m{M}_n = \begin{pmatrix}
\m{M}_{n-1} & \phi_*(x_n) \\ \phi_n(x_*) & \phi_n(x_n)
\end{pmatrix}
\]
and use Schur complement formula for its determinant to obtain
\[
\det \m{M}_n = \det \m{M}_{n-1} \del[2]{\phi_{n}(x_n) - \sum_{j=1}^{n-1}\alpha_j \phi_{j}(x_n)}
\]
where the coefficients $\alpha_j$ are defined as
\[
\del{\alpha_1, \ldots, \alpha_{n-1}} = \del{\phi_n(x_1), \ldots, \phi_{n}(x_{n-1})} \m{M}_{n-1}^{-1}
\]
and in particular do not depend on $x_n$.
If the statement $\det \m{M}_n = 0$ for all $x_n \in \c{S}$ were true, it would imply $\phi_{n}(x_n) - \sum_{j=1}^{n-1}\alpha_j \phi_{j}(x_n) = 0$ for all $x_n \in \c{S}$, which would contradict linear independence of $\phi_j$.
The claim follows.
\end{proof}

\ThmFundamentalSpan*

\begin{proof}
We have
\[
K(x, x_k)
=
\sum_{j=1}^N \overline{\phi_j(x_k)} \phi_j(x).
\]
Notice that
\[
\begin{pmatrix}
K(x, x_1) \\
\vdots \\
K(x, x_N) \\
\end{pmatrix}
=
\m{M}_N^*
\begin{pmatrix}
\phi_1(x) \\
\vdots \\
\phi_N(x) \\
\end{pmatrix}.
\]
Hence, by non-degeneracy of $\m{M}_N$, it follows that $K(\., x_j)$ is a basis of $\Span \phi_j$.
\end{proof}

\ThmLaplaceEigenfunctionsGroup*

\begin{proof}
Since $G$ is equipped with the metric induced by the Killing form, the \emph{left shift} and \emph{right shift} maps
\[
&L_g: G \-> G
&
&L_g u = g \bdot u
&
&R_g: G \-> G
&
&R_g u = u \bdot g^{-1}
\]
are Riemannian isometries of $G$.
Let $\Delta$ be the Laplace--Beltrami operator on $G$, and recall that this operator commutes with Riemannian isometries.
Thus, $\Delta$ commutes with $L_g$ and $R_g$ for all $g \in G$.
Let $\phi_{j k} = \Delta \pi^{(\lambda)}_{j k}$.
Then $(\Delta L_g \pi^{(\lambda)}_{j k})(u) = (L_g \phi_{j k})(u) = \phi_{j k}(g \bdot u)$ and, on the other hand,
\[
(\Delta L_g \pi^{(\lambda)}_{j k})(u)
=
\sum_{l=1}^{d_{\lambda}}\pi^{(\lambda)}_{j l}(g) (\Delta \pi^{(\lambda)}_{l k})(u)
=
\sum_{l=1}^{d_{\lambda}}\pi^{(\lambda)}_{j l}(g) \phi_{l k}(u).
\]
Hence $\phi_{j k}(g u) = \sum_{l=1}^{d_{\lambda}}\pi^{(\lambda)}_{j l}(g) \phi_{l k}(u)$.
Denoting $\phi(g) = \del{\phi_{j k}(g)}$ we get $\phi(g \bdot u) = \pi^{(\lambda)}(g) \phi(u)$.
Examining the action for the right shifts in the same manner, we get $\phi(u \bdot g^{-1}) = \phi(u) \pi^{(\lambda)}(g^{-1})$.
Because of this, $\pi^{(\lambda)}(g) \phi(e) = \phi(g) = \phi(e) \pi^{(\lambda)}(g)$ for all $g \in G$, which, by Schur's Lemma, implies $\phi(e) = \alpha_{\lambda} I$ and $\phi(g) = \alpha_{\lambda} \pi^{(\lambda)}(g)$.
This proves that every $\pi^{(\lambda)}_{j k}$ is an eigenfunction of $\Delta$ corresponding to the same eigenvalue $\alpha_{\lambda}$.
By general theory of the Laplacian we have $\alpha_{\lambda} \leq 0$.

This shows that $\cbr[1]{\sqrt{d_{\lambda}}\pi^{(\lambda)}_{j k}}_{\lambda, j, k}$ is an orthonormal basis of eigenfunctions of $\Delta$ in $L^2(G)$, where the Haar measure on $G$ coincides with the Riemannian volume measure on $G$.
\end{proof}

\ThmLaplaceEigenfunctionsHomogeneous*

\begin{proof}
The metric structure of a homogeneous space $\c{M}=G/H$ with $G$ a compact Lie group is inherited from the group $G$ and thus the Laplacian $\Delta_{G/H}$ on $G/H$ satisfies
\[
(\Delta_{G/H} f)(g\bdot H) = \del[1]{\Delta_{G} \tilde{f}}(g) 
\]
where $\Delta_{G}$ is the Laplacian on $G$, $\tilde{f}(g) = f(\phi(g))$ with $\phi: G \-> G/H$ and $\phi(g) = g \bdot H$.

Consider the spherical functions $\pi^{(\lambda)}_{j k}$ for $1 \leq j \leq d_{\lambda}$, $1 \leq k \leq r_\lambda$.
The functions $\pi^{(\lambda)}_{j k}$, as functions on $G$, are matrix coefficients of irreducible unitary representations of $G$ under a suitable choice of basis in the representation space $V_{\lambda}$.
Because of this, the argument from the previous section implies that they are eigenfunctions of $\Delta_{G/H}$ corresponding to the eigenvalues $\alpha_{\lambda}$.
Thus $\cbr[1]{\sqrt{d_{\lambda}}\pi^{(\lambda)}_{j k}}_{\lambda, j, k}$ is an orthonormal basis of eigenfunctions of $\Delta_{G/H}$ in the space $L^2(G/H)$.
Note that, again, the invariant measure measure on $G/H$ coincides with the Riemannian volume measure on $G/H$.
\end{proof}

\begin{result}
Euclidean Matérn kernels can be written in integral form as
\[
k_{\nu, \kappa, \sigma^2}(\v{x}, \v{x}')
&=
\frac{(2\nu)^{\nu}}{\Gamma(\nu)\kappa^{2\nu}}
\int_0^{\infty}
u^{\nu - 1}
e^{-\frac{2 \nu}{\kappa^2} u}
k_{\infty, \sqrt{2 u}, \sigma^2}(\v{x}, \v{x}')
\d u
\\
&=
\sigma^2 \frac{(2\nu)^{\nu} (4 \pi)^{n/2}}{\Gamma(\nu)\kappa^{2\nu}}
\int_0^{\infty}
u^{\nu - 1 + n/2}
e^{-\frac{2 \nu}{\kappa^2} u}
\c{P}(u, \v{x}, \v{x}')
\d u
.
\]
\end{result}

\begin{proof}
This is a proof from \textcite{jaquier2022}, which we provide for completeness.
By \textcite[Section 3.326, Item 2]{gradshteyn2014} we have
\[ \label{eqn:exponential_integral}
\int_0^{\infty} u^m e^{- a u} \d u = \Gamma(m+1) a^{-m-1}
.
\]
Substituting $m = \nu + n/2 - 1$ and $a = 2 \nu / \kappa^2 + 4 \pi^2 \lambda^2$ into this equation and then performing a simple rearrangement of terms, we get
\[
\del{\frac{2 \nu}{\kappa^2} + 4 \pi^2 \lambda^2}^{-\nu - \frac{n}{2}}
&=
\Gamma(\nu + n/2)^{-1}
\int_0^{\infty}
u^{\nu + \frac{n}{2}-1}
e^{-\frac{2 \nu}{\kappa^2} u}
e^{-4 \pi^2 \lambda^2 u}
\d u,
\\
&=
\del{4 \pi}^{-\frac{n}{2}}\Gamma(\nu + n/2)^{-1}
\int_0^{\infty}
u^{\nu - 1}
e^{-\frac{2 \nu}{\kappa^2} u}
(4 \pi u)^{\frac{n}{2}} e^{-4 \pi^2 \lambda^2 u}
\d u,
\\ \label{eqn:aux_sd_formula_eucl}
&=
\sigma^{-2} \del{4 \pi}^{-\frac{n}{2}}\Gamma(\nu + n/2)^{-1}
\int_0^{\infty}
u^{\nu - 1}
e^{-\frac{2 \nu}{\kappa^2} u}
S_{\infty, \sqrt{2 u}, \sigma^2}(\lambda)
\d u
\]
where
\[
S_{\infty, \kappa, \sigma^2}(\lambda)
=
\sigma^2
(2 \pi \kappa^2)^{n/2}
e^{- 2 \pi^2 \kappa^2 \lambda^2}
\]
is the spectral density of the squared exponential kernel, namely
\[ \label{eqn:se_spectral}
k_{\infty, \kappa, \sigma^2}(\v{x}, \v{x}')
=
\int_{\R^n}
S_{\infty, \kappa, \sigma^2}(\norm{\v{\xi}})
e^{2 \pi i \innerprod{\v{x}-\v{x}'}{\v{\xi}}}
\d \v{\xi}.
\]

Now, using the spectral representation of the Matérn kernel together with~\Cref{eqn:aux_sd_formula_eucl}, we write
\[
&k_{\nu, \kappa, \sigma^2}(\v{x}, \v{x}')
=
\frac{\sigma^2}{C_{\nu}}
\int_{\R^n}
\del{\frac{2 \nu}{\kappa^2} + 4 \pi^2 \norm{\v{\xi}}^2}^{-\nu - \frac{n}{2}}
e^{2 \pi i \innerprod{\v{x}-\v{x}'}{\v{\xi}}}
\d \v{\xi}
\\
&=
\frac{1}{C_{\nu} \del{4 \pi}^{\frac{n}{2}}\Gamma(\nu + n/2)}
\int_{\R^n}
\int_0^{\infty}
u^{\nu - 1}
e^{-\frac{2 \nu}{\kappa^2} u}
S_{\infty, \sqrt{2 u}, \sigma^2}(\norm{\v{\xi}})
\d u
\,
e^{2 \pi i \innerprod{\v{x}-\v{x}'}{\v{\xi}}}
d \v{\xi} = \ldots
\]
By changing the order of integration, rearranging terms and using formula~\Cref{eqn:se_spectral} we get
\[ \label{eqn:matern_integral_formula_full}
\ldots =
\frac{1}{C_{\nu} \del{4 \pi}^{\frac{n}{2}}\Gamma(\nu + n/2)}
\int_0^{\infty}
u^{\nu - 1}
e^{-\frac{2 \nu}{\kappa^2} u}
k_{\infty, \sqrt{2 u}, \sigma^2}(\v{x}, \v{x}')
\d u.
\]
Using the fact that $C_\nu = \frac{\Gamma(\nu) \kappa^{2 \nu}}{2^n \pi^{n/2} \Gamma(\nu+n/2) \del{2 \nu}^{\nu}}$ and
noticing that
\[
\frac{1}{C_{\nu} \del{4 \pi}^{\frac{n}{2}}\Gamma(\nu + n/2)}
\frac{2^n \pi^{n/2} \Gamma(\nu+n/2) \del{2 \nu}^{\nu}}{\Gamma(\nu) \kappa^{2 \nu}}
=
\frac{(2\nu)^{\nu}}{\Gamma(\nu)\kappa^{2\nu}}
\]
proves the claim.
\end{proof}

\begin{result}
Euclidean Matérn kernels, defined in the integral sense, are positive semi-definite.
\end{result}

\begin{proof}
This is a proof from \textcite{jaquier2022}, which we provide for completeness.
Assume $k_{\nu, \kappa, \sigma^2}$ is defined by~\Cref{dfn:matern} and that the corresponding heat kernel $k_{\infty, \kappa, \sigma^2}$ is positive (semi)definite.
Take some locations $x_1, \ldots, x_m = \v{x}$ and consider the matrix $\m{K}^{\nu}_{\v{x} \v{x}}$ with elements $k_{\nu, \kappa, \sigma^2}(x_i, x_j)$.
In order to prove that $k_{\nu, \kappa, \sigma^2}$ is positive (semi)definite we need to show that $\m{K}^{\nu}_{\v{x} \v{x}}$ is a positive (semi)definite matrix for an arbitrary choice of $m$ and $\v{x}$.
This means that we need to show that $\v{y}^{\top} \m{K}^{\nu}_{\v{x} \v{x}} \v{y}^{\top} > 0$ for all nonzero vectors $\v{y} \in \R^m$.

To prove that $\m{K}^{\nu}_{\v{x} \v{x}}$ is positive definite, consider the matrices $\m{K}^{\infty, \kappa}_{\v{x} \v{x}}$ with elements $k_{\infty, \kappa, \sigma^2}(x_i, x_j)$.
Then, extending equation~\Cref{eqn:matern_dfn} to matrices, we have
\[
\m{K}^{\nu}_{\v{x} \v{x}}
=
C
\int_0^{\infty}
  u^{\nu - 1 + n/2}
  e^{-\frac{2 \nu}{\kappa^2} u}
  K^{\infty, \sqrt{2 u}}_{\v{x} \v{x}}
  \d u.
\]
Because of this, we obtain
\[
\v{y}^{\top} \m{K}^{\nu}_{\v{x} \v{x}} \v{y}
=
C
\int_0^{\infty}
  \underbracket[0.1ex]{
  u^{\nu - 1 + n/2}
  e^{-\frac{2 \nu}{\kappa^2} u}
  \vphantom{\v{y}^{\top} K^{\infty, \sqrt{2 u}}_{\v{x} \v{x}} \v{y}}
  }_{*}
  \underbracket[0.1ex]{\v{y}^{\top} K^{\infty, \sqrt{2 u}}_{\v{x} \v{x}} \v{y}}_{**}
  \d u
\]
where the factor $*$ of the integrand is obviously positive and the factor $**$ of the integrand is positive  (non-negative) because $\m{K}^{\infty, \sqrt{2 u}}_{\v{x} \v{x}}$ is positive (semi)definite by assumption, thus the integral is positive (non-negative).
Thus $k_{\nu, \kappa, \sigma^2}$ is positive (semi)definite.
\end{proof}

\ThmMaternLie*

\begin{proof}
Substituting the expression~\Cref{eqn:heat_lie} for heat kernels on a compact Lie group $G$ into \Cref{eqn:matern_integral_formula_2}, we get, for $x, y \in G$,
\[
k_{\nu, \kappa, \sigma^2}(g_1 \bdot H, g_2 \bdot H)
&=
\frac{\sigma^2}{C_{\nu, \kappa}}
\int_0^{\infty}
u^{\nu - 1 + n/2}
e^{-\frac{2 \nu}{\kappa^2} u}
\c{P}(u, g_1 \bdot H, g_2 \bdot H)
\d u
\\
&=
\frac{\sigma^2}{C_{\nu, \kappa}}
\int_0^{\infty}
u^{\nu - 1 + n/2}
e^{-\frac{2 \nu}{\kappa^2} u}
\sum_{\lambda \in \Lambda}
e^{-\alpha_{\lambda} u}
d_{\lambda}
\chi_{\lambda}(g_2^{-1} \bdot g_1)
\d u.
\]
Rearranging the terms we get
\[
k_{\nu, \kappa, \sigma^2}(g_1 \bdot H, g_2 \bdot H)
=
\frac{\sigma^2}{C_{\nu, \kappa}}
\sum_{\lambda \in \Lambda}
\ubr{
\int_0^{\infty}
u^{\nu - 1 + n/2}
e^{-\frac{2 \nu}{\kappa^2} u}
e^{-\alpha_{\lambda} u}
\d u
}_{\Phi_{\kappa, \nu+n/2}(a_{\lambda})}
\,\,
d_{\lambda}
\chi_{\lambda}(g_2^{-1} \bdot g_1)
\]
where $\Phi_{\kappa, \nu+n/2}(a_{\lambda})$ is defined as the inner integral.
By \textcite[Section 3.326, Item 2]{gradshteyn2014}, we have
\[
\int_0^{\infty} u^n e^{- a u} \d u = \Gamma(n+1) a^{-n-1}
.
\]
Hence, $\Phi_{\kappa, \nu+n/2}(a_{\lambda}) = \Gamma(\nu+n/2)\del{\frac{2 \nu}{\kappa^2} + a_{\lambda}}^{-\nu-n/2}$. 
We obtain
\[
k_{\nu, \kappa, \sigma^2}(g_1 \bdot H, g_2 \bdot H)
=
\frac{\sigma^2}{C_{\nu, \kappa}'}
\sum_{\lambda \in \Lambda}
\del{\frac{2 \nu}{\kappa^2} + a_{\lambda}}^{-\nu-n/2}
\,\,
d_{\lambda}
\chi_{\lambda}(g_2^{-1} \bdot g_1)
\]
where $C_{\nu, \kappa}'$ is a normalizing constant ensuring that $k_{\nu, \kappa, \sigma^2}(x, x) = \sigma^2$.
\end{proof}

\ThmSmoothnessComp*

\begin{proof}
See~\Cref{appdx:proofs}.
First, note that
\[
\min\del{1, 2 \nu / \kappa^2}
(1 + a_{\lambda})
\leq
\frac{2 \nu}{\kappa^2} + a_{\lambda}
\leq 
\max\del{1, 2 \nu / \kappa^2}
(1 + a_{\lambda}).
\]
Using this, by change of metric expressions it suffices to study $k_{\nu, \kappa, \sigma^2}$ only in the special case when $2\nu/\kappa^2 = 1$.

Recall that $\cbr[1]{\sqrt{d_{\lambda}}\pi^{(\lambda)}_{j k}}_{\lambda, j, k}$ is an orthonormal basis of eigenfunctions of the Laplace--Beltrami operator.
Then under our assumption that $2\nu/\kappa^2 = 1$, by \textcite[Proposition 2]{devito2019}, we see that $k_{\nu, \kappa, \sigma^2}$ is the reproducing kernel of $H^{\nu+n/2}$, and therefore an element of this space.
Therefore, positive definiteness and continuous differentiability follow from one of the Sobolev embedding theorems---see \textcite[Remark 4]{devito2019}. 
The claim follows.
\end{proof}

\ThmTruncationBounds*

\begin{proof}
To begin, write
\[
\norm{k_{\nu, \kappa} - k_{\nu, \kappa}^J}_{L^2(X \x X)}^2
=
\norm[3]{\sum_{j=J+1}^\infty S_{\nu, \kappa}(\lambda_j) f_j(\.) f_j(\.')}^2_{L^2(X \x X)}
=
\sum_{j=J+1}^\infty S_{\nu, \kappa}(\lambda_j)^2 
\]
where we have used the fact that $f_j$ are orthonormal in $L^2(X)$.
Since the function $S_{\nu, \kappa}(\cdot)$ is monotonically decreasing, Weyl's law given by~\Cref{eqn:weyls_law} implies
\[
\sum_{j=J+1}^\infty S_{\nu, \kappa}(\lambda_j)^2
\leq
\sum_{j=J+1}^\infty S_{\nu, \kappa}(C^{-1} j^{2/n})^2.
\]

First, assume $\nu < \infty$.
Using elementary algebra and the Maclaurin--Cauchy test, we have
\[
\sum_{j=J+1}^\infty
S_{\nu, \kappa}(C^{-1} j^{2/n})^2
=
\sum_{j=J+1}^\infty
\del{\frac{2 \nu}{\kappa^2} + C^{-1} j^{2/n}}^{-2\nu-n}
&\leq
C^{2 \nu + n}
\sum_{j=J+1}^\infty
j^{-4\nu/n-2}
\\
&\leq
C^{2 \nu + n}
\int_{J}^\infty
u^{-4\nu/n-2} \d u
\\
&=
\underbrace{\frac{C^{2 \nu + n}}{4\nu/n+1}}_{\widetilde{C}^2}
J^{-4\nu/n-1}.
\]
This proves the claim for $\nu < \infty$.
Now, assume $\nu = \infty$.
First, notice the identity
\[
\int_{J}^{\infty} e^{-a u^b} \d u
=
a^{-1/b} b^{-1}
\int_{a J^b} v^{\frac{1-b}{b}} e^{-v} \d v
\]
which holds by change of variables: $a u^b = v$.

Using the Maclaurin--Cauchy test and the variable change above with $a = \frac{\kappa^2}{C}$ and $b = 2/n$, write
\[
\sum_{j=J+1}^\infty
S_{\nu, \kappa}(C^{-1} j^{2/n})^2
&=
\sum_{j=J+1}^\infty
e^{-\frac{\kappa^2}{C} j^{2/n}}
\leq
\int_{J}^{\infty}
e^{-\frac{\kappa^2}{C} u^{2/n}} \d u
\\
&=
\del{\frac{\kappa^2}{C}}^{-n/2}
\frac{n}{2}
\int_{\frac{\kappa^2}{C} J^{2/n}}
v^{\frac{n-2}{2}} e^{-v} \d v
\\
&=
\del{\frac{\kappa^2}{C}}^{-n/2}
\frac{n}{2}
\Gamma\del{\frac{n}{2}, \frac{\kappa^2}{C} J^{2/n}}
\]
where $\Gamma(\cdot, \cdot)$ is the incomplete Gamma function. \textcite[\href{https://dlmf.nist.gov/8.11}{(8.11)}]{NIST:DLMF} immediately implies
\[
\frac{\Gamma(x, y)}{y^{x-1}e^{-y}} \-> 1
&&
y \-> \infty.
\]
Hence, for large enough $J$, we have
\[
\sum_{j=J+1}^\infty
S_{\infty, \kappa}(C^{-1} j^{2/n})^2
\leq
\del{\frac{\kappa^2}{C}}^{-n/2}
\frac{n}{2}
\Gamma\del{\frac{n}{2}, \frac{\kappa^2}{C} J^{2/n}}
&\leq
C'
\del{\frac{\kappa^2 J^{2/n}}{C}}^{n/2-1} e^{- \frac{\kappa^2 J^{2/n}}{C}}
\\
&=
\widetilde{C}^2
J^{1-2/n} e^{- \frac{\kappa^2 J^{2/n}}{C}}.
\]
The claim follows.
\end{proof}

\printbibliography

\end{document}